\documentclass[11pt,a4paper] {amsart}

\usepackage{amsmath,amscd,amssymb,amsthm,amstext,mathrsfs,dcolumn,mathdots,
booktabs,enumerate}

\usepackage[all]{xy}

\theoremstyle{plain}
   \newtheorem{theorem}{Theorem}[section]
   \newtheorem{proposition}[theorem]{Proposition}
   \newtheorem{lemma}[theorem]{Lemma}
   \newtheorem{corollary}[theorem]{Corollary}
   \theoremstyle{definition}
   
   \newtheorem{definition}[theorem]{Definition}
   \newtheorem{example}[theorem]{Example}
   \theoremstyle{remark}
   
   \newtheorem{remark}[theorem]{Remark}

\newcommand{\Jac}{\operatorname{Jac}}
\newcommand{\Pic}[1]{\operatorname{Pic}{#1}}
\newcommand{\Picu}[1]{\operatorname{\underline{Pic}}{#1}}
\newcommand{\Picf}[1]{\operatorname{Pic}^0{#1}}
\newcommand{\PD}{\mathscr {P}_d}
\newcommand{\OPD}{\overline{\mathscr {P}_d}}
\newcommand{\LL}{\overline{\mathscr {L}}}

\newcommand{\RR}{{\mathbb R}}
\newcommand{\ZZ}{{\mathbb Z}}
\newcommand{\ZTWO}{{\mathbb Z}/2{\mathbb Z}}
\newcommand{\CC}{{\mathbb C}}
\newcommand{\PP}{{\mathbb P}}
\newcommand{\NN}{{\mathbb N}}

\newcommand{\bb}[1]{\mathbb{#1}}
\newcommand\calo{\mathcal O}
\newcommand{\mc}[1]{\mathcal{#1}}
\newcommand{\mcr}[1]{\mathscr{#1}}
\newcommand\op[1]{\operatorname{#1}}

\newcommand{\pp}{\frac{\partial \overline{\partial}}{2\pi {\rm i}} }
\newcommand{\ddiv}[1]{\operatorname{div}(#1)}
\usepackage{color}
\usepackage{microtype}

\begin{document}

\title[]
{K\"ahler quantization of vortex moduli}
\author{Dennis Eriksson}
\address{Matematiska Vetenskaper, Chalmers Tekniska H\"ogskola \& G\"oteborgs Universitet, 41296 G\"oteborg, Sweden}
\email{dener@chalmers.se}
\author{Nuno M. Rom\~ao}
\address{Institut f\"ur Mathematik, Universit\"at Augsburg, 86135 Augsburg, Germany}
\email{Nuno.Romao@math.uni-augsburg.de}
\date{\today}
\begin{abstract}
We discuss the K\"ahler quantization of moduli spaces of vortices in line bundles over compact surfaces $\Sigma$. This furnishes a semiclassical framework for the study of quantum  vortex dynamics in the Schr\"odinger--Chern--Simons model. We employ Deligne's approach to Quillen's metric in determinants of cohomology to construct all the quantum Hilbert spaces in this context. An alternative description of the quantum wavesections, in terms of multiparticle states of
spinors on $\Sigma$ itself (valued in a prequantization of a multiple of its area form), is also obtained. This viewpoint sheds light on the nature of the quantum solitonic particles that emerge from the gauge theory. We find that in some cases (where the area of $\Sigma$ is small enough in relation to its genus) the dimensions of the quantum Hilbert spaces may be sensitive to the input data required by the quantization scheme,  and also address the issue of relating different choices of such data geometrically.
\end{abstract}
\maketitle
\setcounter{tocdepth}{1}

\tableofcontents

\section{Introduction}

This article is concerned with an application of geometric quantization to a gauge-theoretic setting, in a similar spirit to the study of moduli of flat connections in relation to
Chern--Simons and conformal field theory~\cite{Hit,AxdPiWi}. As phase spaces, we consider moduli spaces of gauged vortices on a compact Riemann surface $\Sigma$. Thus we go beyond the
setting of pure gauge theory and incorporate the coupling of the gauge field $A$ (a connection in a principal $G$-bundle $P\rightarrow \Sigma$) to a matter field $\phi:\Sigma \rightarrow P\times_G X$ (a section of an associated bundle with typical K\"ahler fibre $X$). We shall simplify the problem drastically by restricting ourselves to Abelian gauge theory and line bundles, setting $G={\rm U}(1)$ and $X=\CC$.
One of our goals is to highlight how this problem draws in some novel aspects that can be perceived as orthogonal to the mainstream research on quantization of the moduli of flat connections.

In principle, a study of geometric quantization like the one we propose to take on is expected to provide at least two outputs. The first one is a description of the quantum Hilbert spaces that result from
the quantization, starting from classical data that may need to be supplemented by extra structures, required as ancillary ingredients in the construction (such as Bohm--Aharonov or statistical phases, polarisations, and metaplectic structures).
The second one is a framework to describe (geometrically, if possible) how the quantum Hilbert space depends on the choice of such extra structures. In this paper, we cater for both of these expected outputs, in the following sense:
\begin{itemize}
\item[(i)] We give a description of the quantum Hilbert spaces as vector spaces directly associated to the algebraic geometry of $\Sigma$, admitting a clean interpretation as spaces of multiparticle states in standard nonrelativistic quantum mechanics; see Theorem~\ref{thmdims}. The resulting picture, as we shall see, has the attractive feature of being consistent with the spin-statistics theorem of quantum mechanics.
\item[(ii)] We show how the various choices of extra structures can be fitted into a geometric family (see equation (\ref{bundleonjac})), in the sense of moduli. In addition, we will argue in Theorem~\ref{noHit} that the main route used to show independence of choices in the geometric quantization of moduli of flat connections, by means of a projectively flat connection over the space of choices (also known as Hitchin's connection), is unsuitable for the general quantization problem at hand.
\end{itemize}

To be more precise, the classical phase space we deal with is the moduli space $(\mathcal{M}_d, \omega_{L^2})$ of solutions of the vortex equations in a Hermitian line
bundle $L\rightarrow \Sigma$ of degree $d$, equipped with the symplectic form
associated to its K\"ahler $L^2$-geometry. The definition of $\omega_{L^2}$ involves the Hermitian metric in $L\rightarrow \Sigma$, a Riemannian metric on $\Sigma$, and a real parameter $\tau$; we refer the reader to Section~\ref{sec:vortices} for background. Under the assumption $\tau \in\,  ] \frac{4\pi d}{{\rm Vol(\Sigma)}},\infty[$, where ${\rm Vol}(\Sigma):=\int_\Sigma \omega_\Sigma$ and $\omega_\Sigma$ is the area form, the moduli space
$\mathcal{M}_d$ is the symmetric product
\begin{equation} \label{modspaces}
\mathcal{M}_d  \cong S^d \Sigma:=\Sigma^d / \mathfrak{S}_d
\end{equation}
with its natural complex structure $J^{j_\Sigma}$, induced from the complex structure $j_\Sigma$ of $\Sigma$. However, the symplectic structure $\omega_{L^2}$ (or equivalently, the associated K\"ahler metric $g_{L^2}$) is much harder to
describe. In the limit $\tau \rightarrow \frac{4\pi d}{{\rm Vol(\Sigma)}} $ (referred to as the regime of dissolving vortices~\cite{ManRom}, corresponding to an asymptotically vanishing section $\phi$ and a connection ${\rm d}_A$ of constant curvature with respect to $\omega_\Sigma$), the
complex structure asymptotically determines $\omega_{L^2}$ as the K\"ahler form of a generalisation of the Bergman geometry~\cite{JosCRS} on $\Sigma$; but in general we must start from classical data that are not explicitly accessible, in contrast with the problem of quantization of moduli of flat connections. We shall show how, for most of our purposes, this difficulty can be surmounted in a satisfactory way through a description of the $L^2$-geometry by means of Quillen's metrics on determinants of cohomology, extending a discussion that was initiated by~\cite{Dey} borrowing inspiration from~\cite{BisRag}.
 
The most natural polarisations to consider in the geometric quantization of a K\"ahler phase space are the complex polarisations compatible with the K\"ahler form. In our context, there are both (a) a natural 
family of such polarisations on $S^d\Sigma$, namely, those determined
by compatible complex structures $J^{\tilde\jmath_\Sigma}$ induced from
complex structures $\tilde\jmath_{\Sigma}$ on the surface
$\Sigma$; and (b) a preferred polarisation in this family, namely,
the one corresponding to the $J^{j_\Sigma}$ determined by the particular complex structure $j_\Sigma$ featuring in one of the vortex equations (see (\ref{vort1}) and (\ref{delbarA}) below), and which can therefore be regarded as part of the classical data. Once one takes heed of this preferred choice of complex structure $J^{j_\Sigma}$, the discussion of dependence of the quantization on complex polarisations is secondary, and perhaps even pointless. What is physically more relevant is understanding the dependence on other type of input, such as metaplectic data that must be fed into the quantization process. As we shall see, a natural parametrisation of the relevant metaplectic corrections is absorbed by the other choice required in our construction: the prequantization of a particular rescaling of the area form $\omega_\Sigma$.

Before we lay out the contents of this paper, we want to provide a broader panorama of the gauge theory context in which our problem naturally fits. The vortex equations (\ref{vort1})--(\ref{vort2}) that we will work with are self-duality equations for a Ginzburg--Landau energy functional of the type
\begin{equation}\label{GinzLand}
{\sf E}_\xi (A,\phi)= \frac{1}{2}\int_\Sigma \left( \left|F_A\right|^2 + \left|{\rm d}_A \phi \right|^2 + \frac{\xi}{4} \left( |\phi|^2 -\tau\right)^2 \right) 
\end{equation}
at critical coupling $\xi \rightarrow  1$, and describe minima of this functional in section homotopy classes associated to degrees $d>0$. There are several field theory
models for vortex dynamics incorporating this functional as a potential energy, the most familiar being the Abelian Higgs model in 1+2 dimensions~\cite{JafTau}. This is a Lorentzian
gauged sigma-model, and for slow velocities its dynamics has been shown to be approximated by the
geodesic flow of the $L^2$-metric on $\mathcal{M}_d$ (see \cite{StuAHM} for the analysis of the case $\Sigma=\CC$ and $d=2$), following an idea sketched by Manton~\cite{ManSutTS, StuAL}. Another model of vortex dynamics, including a linear combination of the Chern--Simons functional with a gauged Schr\"odinger term as kinetic energy, was introduced in~\cite{ManFVD} and studied further in~\cite{RomSpe, KruSut}. Its kinetic term is closely related to the Chern--Simons--Dirac functional used by Kronheimer and Mrowka in the study of Seiberg--Witten theory on 3-manifolds~\cite{KroMro}. For values of the coupling constant $\xi$ close, but not equal to $1$, the functional (\ref{GinzLand}) induces~\cite{RomPhD} a Hamiltonian dynamical system on the moduli space $\mathcal{M}_d$, where now the $L^2$-geometry supplies the symplectic structure $\omega_{L^2}$ (rather than the Riemannian structure $g_{L^2}$ for canonical dynamics in ${\rm T}^*\mathcal{M}_d$ as in the Abelian Higgs model); the hope~\cite{ManFVD,RomPhD} that this Hamiltonian system should approximate slow Schr\"odinger--Chern--Simons vortex dynamics has been substantiated in reference~\cite{DemStu}. For
applications of this model in condensed matter physics, see e.g.~\cite{TonTur}.

The geometric quantization of $(\mathcal{M}_d, \omega_{L^2})$ is to be interpreted as a semiclassical approximation (involving a truncation to low energies) of the quantum Hilbert space associated to the
quantum version of the Schr\"odinger--Chern--Simons model of~\cite{ManFVD}; this viewpoint was discussed
in~\cite{RomQVS} in the case $\Sigma=S^2$. From a functorial field theory perspective, the underlying quantization is dealing strictly with product spacetimes of the form $I\times \Sigma$, where $I\subset \RR$ is a time interval; but it could in principle be extended to more general bordisms, taking advantage of the interpretation of Seiberg--Witten moduli spaces as canonical relations~\cite{BatWeiGQ} between vortex moduli spaces (see \cite{DonMT, Ngu}). Upgrading the quantization produced in this paper to a quantum field theory in Lagrangian formulation would require the construction of an Atiyah--Segal functor on (a suitable class of) Riemannian bordisms with necks, so as to satisfy the Atiyah--Segal axioms~\cite{Ati,Seg}. At the moment, this remains a considerable challenge.

Let us now briefly summarise the organisation of this paper. In Section~\ref{sec:vortices} we provide information about the moduli spaces of vortices in line bundles on a compact surface, seen as a family of K\"ahler manifolds. Section~\ref{sec:geomquant} is a short review of K\"ahler quantization with the purpose of fixing our basic terminology.
In Section~\ref{sec:univ}, we describe how the $L^2$-geometry on the spaces ${\mathcal{M}}_d$ is captured by fibre integration formulas; ultimately, this viewpoint justifies the relevance of Quillen's metrics to our context. Section~\ref{sec:Quil} is a summary of definitions and basic techniques related to determinants of cohomology in families of curves and the associated natural metrics; our perspective focuses on the pairing of line bundles introduced by Deligne~\cite{determinant}, and is geared towards the application in this article. Using these tools, we construct all K\"ahler quantizations of the $L^2$-geometry of $\mathcal{M}_d$ in Section~\ref{sec:qdata}, formalising
the input data required in terms of the geometry given on the surface $\Sigma$  at the classical level. Section~\ref{sec:dims} furnishes a description of the quantum wavesections in terms of objects defined on the surface $\Sigma$; the
main outcome is that vortices in line bundles of degree $d$ quantize as states of $d$ fermionic particles on the surface --- each individual fermion being a spinor valued in a prequantisation of $\frac{\tau}{2} \omega_\Sigma$, where the real parameter $\tau$ defines the vacuum of the Higgs field $\phi$ in the field theory model. Finally,
in Section~\ref{sec:projflat} we address the problem of relating different quantizations in our scheme, showing that the main tool used in the
K\"ahler quantization of moduli spaces of flat connections does not
apply to our problem in nontrivial cases of genus $g>1$.\\

\noindent
{\bf Acknowledgements:}
{This project was started as part of the activities of a Junior Trimester Program on ``Mathematical Physics'' hosted at the Hausdorff Research
Institute for Mathematics (HIM), University of Bonn, in 2012. We would like to thank HIM for hospitality, as well as Marcel B\"okstedt (Aarhus), Kai Cieliebak (Augsburg), Daniel Huybrechts (Bonn) and Nick Manton (Cambridge) for useful comments.}

\section{Vortices in line bundles} \label{sec:vortices}

Throughout, $\Sigma$ will always denote a compact orientable surface of genus $g$.
We recall how the symmetric products $S^d \Sigma$ in (\ref{modspaces}) realise moduli spaces for the vortex equations in Hermitian line bundles $L\rightarrow \Sigma$ of degree $d$, and how to portray their K\"ahler $L^2$-geometry.

\subsection{The vortex equations}

We fix a K\"ahler structure $(\Sigma, j_\Sigma, \omega_\Sigma)$ on the surface $\Sigma$. The Hodge star-operator of the underlying metric
 $g_\Sigma:= \omega_\Sigma(\cdot, j_\Sigma \cdot)$, acting on differential forms, will be denoted by $\ast$.
 
Let $L \rightarrow \Sigma$ be
a complex line bundle of degree $d>0$, endowed with a Hermitian metric $\langle \cdot, \cdot\rangle$ which we take to be $\mathbb{C}$-linear in the first argument. The notation we shall follow throughout is that $\overline{V}\rightarrow M$ stands for a vector bundle $V\rightarrow M$
which has been given a Hermitian structure, so we could as well write
$\overline{L} \rightarrow \Sigma$ for emphasis. Sometimes, our vector bundles will also carry holomorphic structures, but these will not be
written explicitly.

A moment map   $ \mathbb{C} \rightarrow \mathfrak{u}(1)^{*}\cong  \mathbb{R}$ 
for the usual Hamiltonian ${\rm U}(1)$-action on $\mathbb{C}$ is prescribed by choosing a real constant $\tau \in \RR$, and it gives rise to a global map $L\rightarrow \mathbb{R}$ pulling back to $\Sigma$ as
\begin{equation}\label{momentm}
\mu \circ \phi=\frac{1}{2}(\langle \phi,\phi \rangle -\tau), \qquad \phi \in\Gamma:= \Gamma(\Sigma,L).
\end{equation}

\begin{definition}
Given the data $(j_\Sigma,\omega_\Sigma,\langle\cdot,\cdot\rangle, \tau)$,  {\em vortices} in $\overline{L}\rightarrow \Sigma$ are pairs $({\rm d}_A, \phi)$ consisting of a
smooth unitary connection ${\rm d}_A$
in  $\overline{L} \rightarrow \Sigma$ with curvature $F_A\in \Omega^2(\Sigma; \mathbb{R})$, and a smooth section $\phi\in \Gamma(\Sigma, L)$, satisfying the {\em vortex equations}~\cite{Bog}
\begin{eqnarray}
&\bar\partial_{A} \phi  = 0, & \label{vort1} \\
&F_{A} + (\mu \circ \phi)\, \omega_{\Sigma} =0. & \label{vort2}
\end{eqnarray}
\end{definition}

The differential operator 
\begin{equation} \label{delbarA}
{\bar \partial}_A:= \frac{1}{2}\left({\rm d}_A+{\rm i}\, {\rm d}_A\circ j_\Sigma\right),
\end{equation}
which satisfies $({\bar \partial}_A)^2=0$,
is the holomorphic structure on $L\rightarrow \Sigma$ determined by ${\rm d}_A$ and the complex structure $j_\Sigma$ on the base~\cite{DonKro}.
Both equations (\ref{vort1}) and (\ref{vort2}) are invariant under the group
of unitary gauge transformations ${\rm Aut}_\Sigma(L)\cong C^\infty(\Sigma,{\rm U}(1))$, acting by
\begin{equation}\label{gauge}
({\rm d}_A,\phi) \mapsto ({\rm d}_A - {\rm i} u^{-1}{\rm d}u,u\phi), \quad
u \in {\rm Aut}_\Sigma(L),
\end{equation}
and one is usually interested in solutions only up to this action. Infinitesimal gauge transformations are described by the induced action of the infinite-dimensional Abelian Lie algebra
$$
{\rm Lie} ({\rm Aut}_\Sigma(L)) \cong C^\infty(\Sigma,\mathfrak{u}(1)) \cong C^\infty (\Sigma,\RR).
$$

Equation (\ref{vort1}) expresses that $\phi$ is a holomorphic section with respect to
the holomorphic structure $\bar\partial_A$.
Therefore, to each solution $({\rm d}_A,\phi)$ we can associate an effective
divisor $(\phi)\in {\rm Div}^+ (\Sigma)$ of
zeroes of $\phi$,
whose degree coincides with the first Chern number
\[
d=  \frac{1}{2 \pi}\int_{\Sigma} F_A = c_1(L)[\Sigma] = {\rm deg}\, L. \]
This integer can be thought of as a quantized magnetic flux in units of $2\pi$ (i.e.\ the total number of vortices in the field configuration), while
the zeroes of $\phi$ specify precise locations for $d$ individual {\em vortex cores} on $\Sigma$, counted with multiplicity. For emphasis, we will
sometimes speak of a {\em $d$-vortex}, and of a {\em multivortex} if $d>1$.

From integrating (\ref{vort2}), it follows that
\begin{equation}\label{L2norm}
|\!|\phi|\!|^2_{L^2}:=\int_{\Sigma} \langle\phi,\phi \rangle \omega_{\Sigma}= \tau {\rm Vol}(\Sigma) - 4 \pi d,
\end{equation}
thus we learn that $(\Sigma,\omega_\Sigma)$ can only
support a $d$-vortex provided
\begin{equation} \label{Brad}
\tau {\rm Vol}(\Sigma) \ge 4 \pi d.
\end{equation}

\subsection{Vortex moduli spaces} \label{sec:moduli}
Conversely, one can prove the following result~\cite{NogYMH, BraVHLB, GarVRS, GarICV}:

\begin{theorem} \label{moduliSd}
Consider a  line bundle $\overline{L} \rightarrow \Sigma$ of degree $d$, equipped with the geometric data introduced above.
Assume that the strict inequality
\begin{equation}\label{Bradlow}
\tau {\rm Vol}(\Sigma) > 4 \pi d
\end{equation}
is satisfied. Then given any effective divisor $D$ of degree $d$ on $\Sigma$, 
one can construct a solution $({\rm d}_A,\phi)$ to equations (\ref{vort1}) and (\ref{vort2}) such that $(\phi)=D$,
and this solution is unique up to the gauge action (\ref{gauge}).
\end{theorem}

Thus once the assumption (\ref{Bradlow}) is made, which we shall do from now on, there is a moduli space $\mathcal{M}_d$ of $d$-vortices up to gauge transformations,
and it can be identified with the symmetric product (\ref{modspaces}) ---
which plays a prominent role in classical algebraic geometry of curves as the space of effective divisors of degree $d$ (see~\cite{arabello}).

A well-known fact from two-dimensional topology (see e.g.~\cite[p.~18]{arabello}) is that the quotient on the right-hand side of (\ref{modspaces}) (where the symmetric group $\mathfrak{S}_d$ acts by permuting the $d$ copies of $\Sigma$)
is smooth, even though the $\mathfrak{S}_d$-action is not free. 

It can be verified~\cite{SamVS,StuAHM} that
the complex structure $J^{j_\Sigma}$ on $S^d \Sigma$ induced by $j_\Sigma$ (see \cite{arabello})
coincides with a natural complex structure on $\mathcal{M}_d$ induced by the (almost) complex structure described as
\begin{equation}\label{cxstr}
J: (\dot A, \dot \phi) \mapsto (*\dot A,{\rm i}\dot\phi),
\end{equation}
on each tangent space ${\rm T}_{({\rm d}_A,\phi)}(\mathcal{A}\times \Gamma)$ of the space of all pairs (or ``fields'')  $ ({\rm d}_A,\phi)$. Here, each tangent space is interpreted as an affine space modelled on the vector space ${\Omega}^1(\Sigma;\RR)\times\Gamma(\Sigma,L) \ni (\dot A,\dot \phi)$. 
The fact that the complex structure (\ref{cxstr}) descends to $\mathcal{M}_d$ follows immediately from the realisation of the moduli space as a K\"ahler quotient, which we explain next.

\subsection{$L^2$-geometry} \label{sec:L2geom}

Given the metric structures on $\Sigma$ and $L$, each infinite-dimensional tangent space  ${\rm T}_{({\rm d_A,\phi})}(\mathcal{A}\times \Gamma)$ is equipped with an $L^2$-inner product
which we normalise in such a way that
\begin{equation}\label{L2metric}
|\!| (\dot A, \dot\phi)|\!|^2_{L^2}:=\frac{1}{4\pi}\int_{\Sigma}\left(\dot{A}\wedge \ast \dot A+\langle \dot{\phi}, \dot{\phi}\rangle \,\omega_{\Sigma}\right).
\end{equation}
This is just a particular case of the $L^2$-inner product on the space of sections of a (tensor product of a) Hermitian vector bundle over a base
endowed with a volume form (coming from a Riemannian structure or a symplectic structure, say), which is familiar from Hodge theory and will be employed in other sections of this paper; its name is meant to emphasise that it provides a generalisation of the usual $L^2$-inner product of functions in real analysis.
Since this inner product is independent of the fields $({\rm d}_A,\phi)$, (\ref{L2metric}) can be regarded as defining a flat Riemannian metric on the infinite-dimensional space
${\mathcal A}\times \Gamma$; it is a K\"ahler metric, as it is (pointwise) compatible with the complex structure (\ref{cxstr}).

To see that this K\"ahler structure induces a K\"ahler structure on the finite-dimensional manifold $\mathcal{M}_d$, one can argue as follows~\cite{GarVRS}. Observe that the first vortex equation (\ref{vort1}) is invariant under the complex structure $(\ref{cxstr})$, and so it defines a K\"ahler submanifold $\mathcal{N}_d\subset
\mathcal{A}\times \Gamma$ on which the gauge group ${\rm Aut}_\Sigma(L)$ acts. Moreover, this action is Hamiltonian, with a moment map for it being
given by the left-hand side of the second vortex equation (\ref{vort2}) --- more precisely, by its image 
under the Hodge star-operator $*$.
We can therefore interpret the moduli space $\mathcal{M}_d$ as an infinite-dimensional
version of a symplectic quotient
\begin{equation} \label{Mdsymplquot}
{\mathcal M}_d =  \mathcal{N}_d /\!\!/ {\rm Aut}_\Sigma(L).
\end{equation}

An alternative interpretation for the quotient (\ref{Mdsymplquot}) is the following. The fact that $\mathcal{N}_d$ is K\"ahler
and that the action of ${\rm Aut}_\Sigma (L)$ is holomorphic and Hamiltonian implies that this action extends uniquely to an action of the complexification
$$
{{\rm Aut}_\Sigma(L)^\CC} \cong {\rm Aut}_\Sigma(L) \times {\rm i}\, {\rm Lie} ({\rm Aut}_\Sigma(L)).
$$
It is then interesting to consider the corresponding space of orbits $\mathcal{N}_d/{\rm Aut}\Sigma(L)^\CC$, but this space is not well behaved; instead, one should replace it by a geometric quotient, retaining only orbits that are {\em stable} in an appropriate sense. In our setting, this means
that one should remove the orbit of the zero-section $\phi=0$, and under this prescription there is an identification of the geometric quotient with the symplectic quotient (\ref{Mdsymplquot}).
Such an identification is sometimes called a {\em Hitchin--Kobayashi correspondence}; it amounts to the existence and uniqueness, inside any orbit of the complexified group through a stable point, of a unique orbit of the original Lie group where the moment map vanishes (see e.g.~\cite{MunHK}).

The quotient (\ref{Mdsymplquot}) receives a symplectic structure which is compatible with the complex structure (\ref{cxstr}); hence the moduli space $\mathcal{M}_d$ is a K\"ahler manifold.
In fact, we obtain in this way a family
$\omega_{L^2, \tau}$ of K\"ahler structures for each $d$, parametrised by $\tau \in \, ]\frac{4 \pi d}{{\rm Vol}(\Sigma)}, \infty[$, reflecting the choice of moment map, but we will henceforth suppress the explicit $\tau$-dependence from our notation and write simply
$\omega_{L^2}$. 

There are other ways of describing the $L^2$-K\"ahler structure on $\mathcal{M}_d(\Sigma)$. One way, which is useful to obtain a localisation formula for the $L^2$-metric~\cite{SamVS}, is based on the
inclusion ${T}_{({\rm d}_A,\phi)} \mathcal{M}_d \subset {\rm T}_{({\rm d}_A,\phi)} \mathcal{A}\times \Gamma$
provided by linearising equations (\ref{vort1})--(\ref{vort2}) about a solution $({\rm d}_A,\phi)$, and resorting to the Coulomb gauge; see \cite{StuAHM}. Another description~\cite{Per},
which gives a more useful vantage point to our
considerations in this paper, uses a universal  bundle over the product $\mathcal{M}_d\times \Sigma$ (the line bundle corresponding to the
universal effective divisor of degree $d$, see~\cite[p~.164]{arabello}); we shall come back to this in Section~\ref{sec:univ}.

\subsection{The K\"ahler classes $[\omega_{L^2}]$} 
\label{sec:Kaehlerclasses}

The following description of the K\"ahler class $[\omega_{L^2}] \in H^2(S^d\Sigma;\mathbb{Z})$ was put forward by Manton and Nasir~\cite{ManNasV} (but see also~\cite{Per}):
\begin{equation} \label{omegacoh}
[\omega_{L^2}]= 2\pi \theta +   \left( \frac{\tau {\rm Vol}(\Sigma)}{2} - 2 \pi d \right) \eta.
\end{equation}
Here, $\theta, \eta \in H^2(S^d\Sigma;\mathbb{Z})$ are integral cohomology classes on the symmetric product spanning the whole K\"ahler cone. To describe these generators, one can resort to the isomorphism
\begin{equation} \label{isoXi}
\Xi:  H_0(\Sigma;\ZZ) \oplus \bigwedge{}^2 H_1(\Sigma;\ZZ)  \stackrel{\cong}{\longrightarrow}  H^2(S^d\Sigma;\ZZ)
\end{equation}
established in Lemma~2.3 of~\cite{Per} for $d>1$. One has $\eta=\Xi([{\rm pt}])$, the image of the obvious generator of the first summand,
and $\theta=\Xi\left(\sum_{j=1}^g a_i\wedge a_{i+g}\right)$, where $\{a_i,a_{i+g}\}_{i=1}^g$ is a symplectic basis for $H_1(\Sigma;\ZZ)$ and $g$ the genus of $\Sigma$. An alternative interpretation
 for $\theta$ is the following (see Proposition~(2.1)(ii) in \cite{BerTha}): it is the pull-back
$\theta=({\rm AJ}_d)^*\Theta$ by the Abel--Jacobi map
${\rm AJ}_d:S^d\Sigma \rightarrow {\rm Jac}(\Sigma)$ of the theta-class $\Theta$ (the latter being the first Chern class of the
line bundle defined by the theta-divisor on the Jacobian ${\rm Jac}(\Sigma)$). It is a well-known fact that 
$
\Theta^g=g!\, {\rm PD}([{\rm pt}])\in H^{2g}({\rm Jac}(\Sigma);\ZZ)
$
(see e.g.~\cite[\S I.5]{arabello}), where $\rm PD$ denotes the map taking the Poincar\'e dual of a cycle, leading to the relation
\begin{equation} \label{thetag}
\theta^g=g! \, ({\rm AJ}_d)^*{\rm PD}([{\rm pt}]) \; \in\; H^{2g}(S^d \Sigma;\ZZ) \quad \text{ whenever }d\ge g
\end{equation}
in the singular cohomology ring of $S^d\Sigma$.

The formula (\ref{omegacoh}) is consistent with an interpretation~\cite{ManRom} of the $L^2$-geometry on the vortex moduli space $\mathcal{M}_d$ as a deformation of the geometry
of line bundles on $(\Sigma,j_\Sigma)$ (encoded by the natural flat K\"ahler metric on the Jacobian) as $\tau$ runs away from the critical value $\frac{4 \pi d}{{\rm Vol}(\Sigma)}$. It also implies that  K\"ahler structures $\omega_{L^2}$ corresponding to different values of $\tau $ (with all the geometric data on $L$ and $\Sigma$ fixed) are not symplectomorphic, and hence also not isometric.

A question that is in a sense complementary to this dependence on $\tau$ becomes relevant in the context of K\"ahler quantization. Suppose that we fix 
a compact oriented surface $\Sigma$ with area form $\omega_\Sigma \in \Omega^2(\Sigma;\RR)$, a Hermitian line bundle $\overline{L}\rightarrow \Sigma$ of degree $d\in \NN$, and  $\tau\in \RR$ such that \eqref{Bradlow} holds. Then, for each choice of
complex structure $j_\Sigma$ compatible with $\omega_\Sigma$, one may define the operator
(\ref{delbarA}), write down the vortex equations (\ref{vort1})--(\ref{vort2}) and consider the corresponding moduli space of vortices, i.e. the K\"ahler manifolds 
\begin{equation} \label{MdKaehler}
\mathcal{M}_d=\left(S^d\Sigma, J^{j_\Sigma}, \omega_{L^2}^{j_\Sigma}\right),
\end{equation}
where we now made the dependence of $\omega_{L^2}$ on $j_\Sigma$
explicit --- recall that $\omega_{L^2}$ is obtained by
symplectic reduction \eqref{Mdsymplquot} of a symplectic space $\mathcal{N}_d$ depending on $j_\Sigma$.
Are the resulting symplectic manifolds $\left(S^d \Sigma, \omega_{L^2}^{j_\Sigma}\right)$ symplectomorphic?

Moser's theorem  yields a positive answer to this question, which we may rephrase in a slightly stronger version:

\begin{lemma} \label{Mosertrick}
Any two complex structures $j_{\Sigma}^0,j_{\Sigma}^1$ on $\Sigma$ compatible with a given area form $\omega_\Sigma \in \Omega^2(\Sigma;\RR)$ give rise to strongly isotopic symplectic structures
$\omega_{L^2}^{j_\Sigma^{0}}, \omega_{L^2}^{j_\Sigma^{1}}$ on the manifold $S^d\Sigma$.
\end{lemma}

\begin{proof}
For the benefit of the reader, we spell out what we mean (see Definition~7.1 in \cite{CdSLSG}) by the two symplectic structures 
$\omega_{L^2}^{j_\Sigma^{0}}, \omega_{L^2}^{j_\Sigma^{1}}$
being strongly isotopic: there exists a differentiable map $\varrho: [0,1]\times S^d \Sigma \rightarrow S^d \Sigma$, interpreted as a path $t\mapsto \varrho_t$ of diffeomorphisms of $S^d\Sigma$, such that
$$\varrho_0={\rm id}_{S^d\Sigma}\qquad \text{ and } \qquad\varrho_1^*\left( \omega_{L^2}^{j_\Sigma^{1}}\right) = \omega_{L^2}^{j_\Sigma^{0}}$$hold.

For a given $\omega_ \Sigma$, we know~\cite[Proposition~4.1(i)]{McDSalST} that the space $\mathcal{J}(\Sigma,\omega_\Sigma)$ of almost complex structures on $\Sigma$ compatible with $\omega_\Sigma$ is contractible, hence path-connected; and that on a surface any almost complex structure is automatically integrable~\cite[Theorem~4.16]{McDSalST}.
So we may choose a  path of complex structures $j:[0,1]\rightarrow \mathcal{J}(\Sigma,\omega_\Sigma)$ with endpoints $j(0)=j_\Sigma^0$ and $j(1)=j_\Sigma^1$. 
For any $0\le t \le 1$, we can follow the recipe in Section~\ref{sec:L2geom} with $j_\Sigma=j(t)$ to obtain a symplectic form $\omega_{L^2}^{j(t)} \in \Omega^2(S^d\Sigma;\RR)$ (keeping all the data other than $j_\Sigma$ the same). In particular,
$\{ \omega_{L^2}^{j(t)}\}_{t\in [0,1]}$ is a family of nondegenerate 2-forms in $S^d\Sigma$. Moreover, the formula (\ref{omegacoh})
shows $[\omega_{L^2}^{j(t)}]$ is independent of $t$. These last two facts show that the assumptions of  Moser's theorem, as stated in \cite[Theorem 7.3]{CdSLSG}, are satisfied, and we infer the existence of maps $\varrho$ as above.
\end{proof}

This lemma implies that making a $j_\Sigma$-dependence
explicit in the notation for $\omega_{L^2}$ as in (\ref{MdKaehler}) is ``superfluous up to strong isotopy equivalence''. However, it is
not guaranteed that any of the diffeomorphisms $\rho_1$ furnished by Moser's trick are $(J^{j^0_\Sigma},J^{j^1_\Sigma})$-holomorphic in the sense that $J^{j^1_\Sigma}\circ {\rm d}\varrho_1 = {\rm d}\varrho_1 \circ J^{j^0_\Sigma}$; in particular, they certainly cannot be if $d=1$.
With this caveat, we will adhere to our previous convention
\begin{equation} \label{nicerMd}
\mathcal{M}_d=\left(S^d\Sigma, J^{j_\Sigma}, \omega_{L^2}\right)
\end{equation}
from now on.

\section{K\"ahler quantization} \label{sec:geomquant}

To set up more of our notation, we  briefly recall the steps involved in K\"ahler quantization, understood as geometric quantization in a complex polarisation~\cite{WooGQ, BatWeiGQ}, and then guide the reader through a few illustrative examples that provide a foretaste of the main results in this paper.

\subsection{Prequantization data} The starting point 
is a symplectic manifold $(M,\omega)$, the classical phase space. The aim of geometric quantization is to upgrade classical observables $f \in C^\infty(M)$ (which multiply through the Poisson bracket of $\omega$) by Hermitian operators acting on an appropriate Hilbert space. 
As a first step, one
seeks to construct a Hermitian
line bundle $\overline{\mathcal{L}} \rightarrow M$, equipped with a unitary connection $\nabla$ such that its curvature 2-form $F_\nabla \in \Omega^2(M; \mathbb{R})$  satisfies the condition
\begin{equation} \label{quant}
F_\nabla=\frac{{1}}{ \hbar} \,\omega
\end{equation}
for a fixed  parameter $\hbar \in \mathbb{R}_{>0}$. The data $(\overline{\mathcal{L}}, \nabla)$ with the properties above are referred to as a \emph{prequantum bundle} and a \emph{prequantum connection}, respectively.

Once the prequantization data $(\overline{\mathcal{L}}, \nabla)$ have been fixed, we can consider the \emph{prequantum Hilbert space} of smooth {\em wavesections} $\mathcal{H} = \{\psi \in \Gamma(M,\mathcal{L}):  \| \psi \|_{L^2} < \infty \}$, with its $L^2$-inner product determined by the Hermitian and symplectic structures. If $f\in C^\infty(M)$, we can define an operator
$\widehat{f} \in {\rm End}(\mathcal{H})$  by $\widehat{f} \psi := -{\rm i} \hbar \nabla_{X_f}\psi+ f \psi$, where $X_f$ denotes the Hamiltonian vector field of $f$ with respect to $\omega$. 
We shall henceforth set $\hbar=1$; 
usually one refers to the situation $\hbar \rightarrow 0$
as a ``semiclassical limit'', but this can still be simulated by considering the sequence of prequantizations
 $(\overline{\mathcal{L}}^{\otimes n},\nabla^{(n)})_{n \in \mathbb{N}}$ of $(M,n \omega)$,
 and then letting $n \rightarrow \infty$.

 It is evident that (\ref{quant}) requires the cohomology class $\left[\frac{1}{2\pi} \omega\right] \in H^2(M;\mathbb{R})$ to be
integral, i.e.~its pairing with any 2-homology class should yield integer values:
\begin{equation} \label{Weil}
\left\langle \left[\frac{1}{2\pi} \omega\right], \sigma\right\rangle \in \mathbb{Z}\qquad \text{for all}\; \sigma \in H_2(M;\mathbb{Z}).
\end{equation}
The requirement (\ref{Weil}) is known as the Weil (pre)-quantization condition.  
In fact, the constraint (\ref{Weil}) is also sufficient for at least one Hermitian line bundle $\overline{\mathcal{L}}\rightarrow M$ with connection $\nabla$ to exist (cf.~\cite{WooGQ}, Proposition~8.3.1), but there may be many possible choices. The ambiguity is parametrised by the space of flat line bundles on $M$, which can be interpreted as possible Aharonov--Bohm phases in the quantization and may be modelled more concretely as
\begin{equation}  \label{flat}
H^1(M,{\rm U}(1))\cong{\rm Hom}(H_1(M;\ZZ),{\rm U}(1)).
\end{equation}

\subsection{K\"ahler polarisations}

The next step is to introduce {\em polarisations} $P$ of $(M,\omega)$; see \cite{WooGQ, BatWeiGQ} for the general definition. By means of this gadget, one
truncates the prequantum Hilbert space to subspaces $\mathcal H_P \subset \mathcal H$ consisting of sections which are constant in half of the variables; this step is meant to implement the choice of representation 
in ordinary quantum mechanics. 

In the 
situation where $M$ has a complex structure $J$ for which $\omega$ is a K\"ahler form, there is a natural \emph{K\"ahler polarisation} $P_J$, which singles out {\em polarised wavesections} $\psi$ satisfying $\nabla_X \psi =0$ for all $X \in \Gamma(M,{\rm T}_J^{0,1}(M))$ in the splitting of the complexified tangent bundle associated to $J$.
This leads to the \emph{quantum Hilbert space} of all $J$-holomorphic, square-integrable sections
$${\mathcal H}_{P_J}:= H_J^0(M,\mathcal{L}) \cap \mathcal{H}.$$
Here, $\mathcal{L}\rightarrow M$ is given a holomorphic structure  by the operator $\overline{\partial}_\nabla  := \nabla_J^{(0,1)}$ defined by composing $\nabla$ with the projection onto  $(0,1)$-forms with respect to $J$ (for short, one says that
$\mathcal L$ is {\em polarised} by $J$). Indeed, since $\omega$ is a $(1,1)$-form, $\overline{\partial}_\nabla^2 =0$ and then $\bar\partial_\nabla$ defines an holomorphic structure by the Newlander--Nirenberg theorem. Conversely, in this setting we can recover the prequantum connection from a Hermitian holomorphic line bundle $\overline{\mathcal{L}}$ as its Chern connection, i.e.~the unique unitary connection $\nabla$ on $\mathcal L$ such that $\nabla^{(0,1)}$ is the holomorphic
structure operator~\cite[p.~73]{GH}.

In the case where $(M, J, \omega)$ is both K\"ahler and compact, and we are given a prequantum bundle $\mathcal L \rightarrow M$, there is a distinguished \emph{holomorphic quantization} for which
\begin{equation} \label{HilbKaeh}
\mathcal{H}_{P_J}=H_J^0(M, \mathcal{L})
\end{equation}
is a finite-dimensional vector space. 
Not all classical observables operate on $H^0(M, \mathcal{L})$ via the prequantum operator recipe given above, but only those $f \in C^\infty(M)$ for which $X_f$ is a Killing  field for the underlying metric.

The mathematical formulation of the principle of superposition (see \cite[p. 17]{Dir}) dictates that quantum states correspond precisely to complex rays in the quantum Hilbert space; in our setting, these are parametrised by points of the projective space $\PP(H^0(M, \mathcal{L}))$ associated to (\ref{HilbKaeh}).

\subsection{Metaplectic corrections} \label{sec:metaplectic}

So far we have equipped the quantum Hilbert space with a $L^2$-structure that makes use of the Liouville volume form $\frac{1}{n!} \omega^n$ on the base $M$, where $n=\frac{1}{2}{\dim_\RR M}$. However, in many examples it is more convenient (in addition to leading to more sensible results) to promote polarised sections to polarised half-forms or $\frac{1}{2}$-densities (see Appendix A in \cite{BatWeiGQ}). These pair without reference to further geometry on the base $M$ --- or, more generally, on the space of leaves of the polarisation.

In the setting of K\"ahler quantization over a compact manifold, this modification is conveniently recast as 
an upgrade of (\ref{HilbKaeh}) to
the \emph{corrected quantum Hilbert space}
$H_J^0 (M, \overline{\mathcal{L}}\otimes K_M^{1/2})$ where $K_M^{1/2}$ denotes a choice of square root (or spin structure) of the canonical bundle, assuming it exists. Thus incorporation of half-forms in K\"ahler quantization is tantamount to introducing a {\em metaplectic structure} --- we refer the reader to the discussion in Section 2 of~\cite{AndGamLau}.

 \subsection{A handful of examples}

At this point, we would like to discuss some prototypical examples, all based on the ingredients of Section~\ref{sec:vortices}.

\begin{example} \label{excurve}
Let $(\Sigma,\omega_\Sigma)$ be a connected compact orientable surface equipped with an area form $\omega_\Sigma \in \Omega^2(\Sigma;\mathbb{R})$, which is necessarily symplectic.
In this case, $H_2(\Sigma;\mathbb{Z})\cong\mathbb{Z}$ is generated by the fundamental class $[\Sigma]$ and the Weil prequantization conditions (\ref{Weil}) amount to
\begin{equation} \label{quantSig}
{\rm Vol}(\Sigma):=\int_\Sigma \omega_\Sigma\quad \in\quad 2 \pi \mathbb{N}.
\end{equation}
This guarantees that prequantizations will exist, but there will be infinitely many whenever the genus $g$ of $\Sigma$ is positive. The space of possible prequantizations is
the $2g$-dimensional torus
\begin{equation} \label{torus}
{\rm Hom}({\pi_1(\Sigma)}, {\rm U}(1))\cong{\rm Hom}({H_1(\Sigma;\mathbb{Z})}, {\rm U}(1))\cong T^{2g}.
\end{equation}
This space parametrises
holonomies of unitary flat connections on a degree zero bundle over $\Sigma$, when reference is made to a  particular prequantization $(\overline{Q},\nabla)$ with first Chern class  $c_1(Q)=\left[  \frac{\omega_\Sigma}{2\pi} \right]=k \, {\rm PD}([{\rm pt}])$, where $k$ is an integer.

In this example,  K\"ahler polarisations are prescribed by  the choice of a complex
structure $j_\Sigma$
on $\Sigma$ compatible with $\omega_\Sigma$. Then the phase space becomes a K\"ahler manifold
$(\Sigma, \omega_\Sigma, j_\Sigma)$. As above, we can also identify the prequantum connections $\nabla$ with holomorphic structures on $\overline{Q} \rightarrow \Sigma$.
In this way, the torus (\ref{torus}) is identified  with the component ${\rm Pic}^k(\Sigma)$ of the Picard group of $(\Sigma,j_\Sigma)$ parametrizing holomorphic line bundles of the same  degree $k$ as
$Q$, or (noncanonically)  with the Jacobian variety ${\rm Jac}(\Sigma)$ parametrizing ambiguities
with respect to a holomorphic structure of reference.
\end{example}

\begin{example} \label{rescale}
In this paper, the primary focus will be on the K\"ahler quantization of  $\mathcal{M}_d=(S^d\Sigma, \omega_{L^2})$. The case $d=1$, where $\mathcal{M}_1\cong S^1\Sigma\cong  \Sigma$,
 reduces to the example we have just discussed --- with the slight difference that one replaces $\omega_\Sigma$ by the K\"ahler structure ${\omega_{L^2}}$
 (which itself depends on the choice of an area form $\omega_\Sigma$), leading to a rescaling of the prequantization condition (\ref{quantSig}) to
 \begin{equation} \label{rescaledWeil}
 \frac{\tau}{2}\, {\rm Vol}(\Sigma) \; \in\; 2 \pi \mathbb{N},
 \end{equation}
 as can be read off from (\ref{omegacoh}). 
\end{example}

\begin{example}
It may seem less obvious at first that Example~\ref{excurve} relates to our problem of quantization of $\mathcal{M}_d$ for arbitrary $d\in \mathbb{N}$
and   $\tau>\frac{4\pi d}{{\rm Vol}(\Sigma)}$, but this turns out to be the case. 

There are a few immediate assertions one can make relating the two classes of examples. To start with, it follows from (\ref{omegacoh}) that the existence of prequantizations  is
still expressed by the rescaled condition (\ref{rescaledWeil}) for any $d$. This was already observed in \cite{RomQVS} in the case where $\Sigma$
is a 2-sphere, but it appears nowhere in the paper~\cite{Dey} or its erratum.
Moreover, it is also true that ambiguities in the choice of prequantum data are still parametrised by the Jacobian ${\rm Jac}(\Sigma)$ (a consequence of
the fact ${\rm Pic}^0(S^d\Sigma) \cong {\rm Jac}(\Sigma)\cong{\rm Pic}^0 (\Sigma)$, cf.~\cite{Pol}).
Another easy observation (already made in the Introduction) is that $S^d\Sigma$ has a natural set of complex structures $J^{j_\Sigma}$ which are in one-to-one correspondence with complex structures $j_\Sigma$ on $\Sigma$; but since a particular $j_\Sigma$ enters the classical data via $\bar\partial_A$, it also provides a canonical choice of complex polarisation $P_{J^{j_\Sigma}}$.

The main results of this paper yield a much less obvious relationship between this example and the very basic Example~\ref{excurve}. At this stage, we can formulate it in the following terms. The first statement (to be justified in Section~\ref{sec:qdata}) is that one can set up K\"ahler quantization of $\mathcal{M}_d$ efficiently in terms of prequantisation data for (the multiple $\frac{\tau}{2} \omega_\Sigma$ of) the area form used in the vortex equation (\ref{vort2}). A second, and in a sense converse statement (which is the main content of Section~\ref{sec:dims}), is
that the wavesections we shall obtain can be recast (in a precise sense) as quantum $d$-particle states defined on $\Sigma$ itself and valued in a prequantisation of $\frac{\tau}{2} \omega_\Sigma$; they also have an automatic spinorial component that amounts to a metaplectic correction
for K\"ahler quantization on $\Sigma$.
\end{example}

\section{Universal bundles and $L^2$-geometry} \label{sec:univ}

A useful viewpoint on the $L^2$-geometry of the moduli space of $d$-vortices is got by considering the product space $\mathcal{M}_d \times \Sigma$, which supports a universal line bundle $\mcr{P}_d$ carrying certain geometric data.

\subsection{A general observation}

Let $\mc G$ be a Lie group with complexification $\mc G^\bb C$ and Lie algebra $\mathfrak g$.
Suppose $\bar L \rightarrow V$ is a holomorphic Hermitian line bundle on a K\"ahler manifold $V$, that $\mc G$ acts freely on $L$ preserving the metric, and that this action
extends as a holomorphic action of $\mc G^\CC$.
Suppose furthermore that the induced $\mc G$-action on $V$ is Hamiltonian with moment map $\mu: V \to \mathfrak g^*$. Assume that there is a globally stable Hitchin--Kobayashi correspondence, in the sense that $M := \mu^{-1}(0)/\mc G \cong V/\mathcal{G}^\bb C$ holds.  Denote by $\nabla^V$  the Chern connection of $L$ on $V$, and let $\nabla^{\mu^{-1} (0)} := i^* \nabla^V$ for the inclusion $i: \mu^{-1}(0) \hookrightarrow V$. Then we can descend $L$  to $M$ in two ways. In the first case, the descent $\check L$ of $L$ inherits a Hermitian metric and a connection $\nabla$. In the second case, it inherits a holomorphic structure.

\begin{lemma}\label{lemma:KHcorr} In the above situation, the inherited connection $\nabla$ is the Chern connection, i.e. the unique unitary connection respecting the holomorphic structure.
\end{lemma}

\begin{proof} Unitarity of $\nabla$ follows from the definition. We need to prove that it respects the holomorphic structure. Let $\pi:V \to M$ and $\pi_0: \mu^{-1}(0) \to M$ denote the projections. Also let $s$ be a local holomorphic frame on $\check L \rightarrow M$, and $\omega$ the corresponding connection 1-form. This is the form whose value at $v \in {\rm T}_{[D]} M := 0 \oplus {\rm T}_{[D]} M \subseteq {\rm T}_D \mu^{-1}(0)  = \mathfrak g \oplus {\rm T}_{[D]} M$ is given by $\nabla_{v} \pi_0^* s/\pi_0^*s $. We need to prove that, for the almost complex structure $I$ on $M$ and any $v \in {\rm T}_D M$, we have $I \omega(I v) = {\rm i} \omega(v)$. It follows basically from definition that $\nabla^{\mu^{-1} (0)} \pi_0^* s/ \pi_0^* s =  i^* \nabla^V \pi^* s/ \pi^* s$ evaluated on ${\rm T}_{[D]} M$. Set $ \nabla^V \pi^* s/ \pi^* s =: \tilde{\omega}$. Also essentially by definition, $\tilde{\omega}(Iv) = {\rm i} \tilde{\omega}(v)$ for any $v \in {\rm T}_D V =  \mathfrak g^{\bb C} \oplus {\rm T}_{[D]} M$. Hence it is enough to prove that ${\rm T}_{[D]} M \subset  {\rm T}_D V $ is invariant under the complex structure on $V$. But this is ensured by the definition of the K\"ahler structure on $M$.
\end{proof}

\subsection{The universal degree $d$ line bundle}
\label{sec:Poincline}

First of all,  $\mathcal{M}_d \times \Sigma \cong S^d \Sigma \times \Sigma$ possesses a natural complex structure $(J^{j_\Sigma},j_\Sigma)$, specified by the complex structure $j_\Sigma$ on the Riemann surface $\Sigma$ alone. For an effective divisor $D$ of degree $d$, denote by $[D]$ the corresponding analytic subset of $\Sigma$ (with multiplicities). Then $S^d \Sigma \times \Sigma$ comes equipped with a
universal divisor 
\begin{equation} \label{univdiv}
\mcr{D}_d:=\bigcup_{D\in S^d \Sigma} \{ D \} \times [D] \subseteq S^d\Sigma  \times \Sigma
\end{equation}
which in turn determines a universal holomorphic line bundle $\mcr{P}_d = \mathcal{O}(\mathcal{D}_d)$ over $S^d\Sigma \times \Sigma $. An alternative description of this  line bundle is got by considering the product $\mathcal{V}  \times L$, where $\mathcal{V} \subset {\mathcal A}\times \Gamma$ is the complex submanifold of
vortices $({\rm d}_A,\phi)$ solving the vortex equations (\ref{vort1})--(\ref{vort2}) in a line bundle $\overline{L}\rightarrow \Sigma$ of degree $d$.  The gauge group
$\mathcal{G}={\rm Aut}_\Sigma(\overline{L})$ acts on both factors
and one can take the space of orbits
\begin{equation} \label{univbdle}
 (\mathcal{V} \times L)/\mathcal{G} = \mcr{P}_d
\end{equation}
as the total space of the universal bundle.
Note that the restriction of $\mcr{P}_d$ to each curve of the form $\{[{\rm d}_A,\phi]\}\times \Sigma \subset \mathcal{M}_d\times \Sigma$ yields a holomorphic
line bundle over ${\Sigma}$, which is isomorphic to $L\rightarrow \Sigma$ equipped with the holomorphic structure associated to the degree $d$ divisor $(\phi)$ in $\Sigma$:
\begin{equation} \label{PonSigma}
\mcr{P}_d|_{\{[{\rm d}_A,\phi]\}\times \Sigma} \cong L.
\end{equation}
This description has the advantage of endowing $\mcr{P}_d$  with a Hermitian structure, which is obtained by using these isomorphisms to pull back to the universal bundle the Hermitian structure that has been fixed in $\overline{L}\rightarrow \Sigma$.

\subsection{Canonical connection and section}

Since our moduli spaces parametrise pairs of connections and sections, there are two objects naturally attached to  \makebox{$\overline{\mcr{P}}_d \rightarrow \mathcal{M}_d \times \Sigma$:}
\begin{itemize}
\item a canonical unitary connection ${\rm d}_\mathcal{A}$ in
$\overline{\mcr{P}}_d$;
\item
a canonical holomorphic section $\Phi \in H^0(\mathcal{M}_d \times \Sigma,\mcr{P}_d)$.
\end{itemize}
The various descriptions of the $L^2$-geometry of $\mathcal{M}_d$ proposed in the literature have resorted (more directly or less directly) to at least one of these two
gadgets in a crucial way.

Canonical connections such as ${\rm d}_\mathcal{A}$ are familiar from other moduli constructions in gauge theory~\cite{DonKro}. In our setting,
${\rm d}_{\mathcal{A}}$ is the restriction to $\mathcal{M}_d \times \Sigma$ of a unitary connection in the bundle ${\rm pr}_2^*\overline{L}/\mathcal{G}$ mentioned above,
which was described by Perutz in~\cite[Sec.~2.2.2]{Per} in terms of two other auxiliary connections: a tautological, $\mathcal{G}$-invariant unitary connection on
\makebox{${\rm pr}_2^*\overline{L} \rightarrow (\mathcal{A}\times \Gamma(\Sigma,L))\times \Sigma $,} and a certain connection 1-form on
$\mathcal{A}\times\Gamma(\Sigma,L)$ related to gauge fixing.
Note that we can employ Lemma \ref{lemma:KHcorr} to our situation by taking $V=\mathcal{N}_d$, and conclude that the canonical connection
is the Chern connection of the universal bundle $\overline{\mcr{P}}_d \rightarrow \mathcal{M}_d \times \Sigma$.

The canonical section $\Phi: \mathcal{M}_d \times \Sigma \rightarrow \mcr{P}_d$  is most conveniently described by means of the model (\ref{univbdle}) as the projection to $\mathcal{G}$-orbits of the
$\mathcal{G}$-equivariant and holomorphic map $\mathcal{V}\times \Sigma \rightarrow \mathcal{V}\times L$ given by
\begin{equation}
\Phi: (({\rm d}_A,\phi),\mathbf{x}) \mapsto  ( ({\rm d}_A,\phi), \phi(\mathbf{x})).
\end{equation}
In analogy to the canonical connection, it follows that $\Phi$ restricts to each curve $\{ [{\rm d}_A,\phi]\} \times \Sigma$ to yield the section $\phi$, making use of (\ref{PonSigma}), and also that its divisor of zeroes is $(\Phi)=\mcr{D}_d$.

\subsection{Fibre integration formulas}

One point of view on how the universal/canonical objects we have introduced determine geometry on $\mathcal{M}_g$ uses the operation of
{fibre integration} with respect to the projection $\mathcal{M}_d\times \Sigma \rightarrow \mathcal{M}_d$ onto the first factor.
At the topological level, this corresponds to the slant product~\cite[p.~280]{Hat}
\begin{equation}
\cdot \setminus \cdot: H^4(\mathcal{M}_d\times \Sigma;\mathbb{R}) \times H_2(\Sigma;\mathbb{R}) \rightarrow H^2(\mathcal{M}_d;\mathbb{R})
\end{equation}
by the fundamental class $[\Sigma] \in H_2(\Sigma;\mathbb{Z})$. This viewpoint leads to the description of the K\"ahler class $[\omega_{L^2}]$ obtained by Perutz~\cite{Per},
from which one recovers the formula (\ref{omegacoh}) propounded by Manton and Nasir.

Remarkably, a refinement of this description of the K\"ahler class associated to the $L^2$-geometry has also been argued to
hold at the level of 2-forms (see Theorem 1 and Equation (27) in \cite{Per}, Proposition~3.2 in \cite{BapL2M}, as well as Proposition~7.1 of \cite{BisSch} for a generalisation):
\begin{equation}\label{L2metricformula}
{\omega_{L^2}}=-\frac{1}{4\pi}\int_{\Sigma} \left(    {\tau}\,  p^* \omega_\Sigma  \wedge F_\mathcal{A} +  F_\mathcal{A} \wedge F_\mathcal{A} \right),
\end{equation}
where $p:\mathcal{M}_d\times \Sigma \rightarrow \Sigma$ is the projection
onto the second factor.
In the following section, we make a short interlude to discuss in more detail the algebraic-geometric setting where
formulas such as \eqref{L2metricformula} emerge.

\section{Deligne's formalism of the Quillen metric} \label{sec:Quil}

In this section, we review the metric originally introduced by Quillen~\cite{Qui} in determinants of cohomology, and the related natural metric on Deligne's parings. General references are \cite{determinant, Sou}, and we refer the reader to these for further details. All the vector bundles we consider in this section carry holomorphic structures.

\subsection{The determinant of the cohomology}

Let $\pi: \mcr C \to S$ be a submersive holomorphic map of quasi-projective complex manifolds, with connected fibres ${\mcr C}_s$ (for $s\in S$) of complex dimension one. We will call such a map a 
{\em family of curves}, and sometimes
will abbreviate it as $\mcr{C}/S$. 

By a relative divisor on $\mcr{C}/S$, we mean a (Cartier) divisor on $\mcr C$ whose intersection with each fibre yields a
divisor of that fibre. 
We let
$K_{\mcr{C}/S}:=K_{\mcr{C}}\otimes \pi^*(K_{S}^{-1})$ denote the relative canonical bundle on ${\mcr C}$ associated to $\pi$. It is equipped with a Hermitian metric whenever (the tangent bundles of) $\mcr C$ and $S$ carry Hermitian metrics themselves, which we shall assume from now on.

A vector bundle $E$ on $\mcr C$ restricts to each fibre as
a vector bundle $E_s \rightarrow {\mcr C}_s$ of the same rank. Then we define a line bundle $\lambda(E) \rightarrow S$ whose fibres at every point $s \in S$ are constructed from sheaf cohomology groups as follows:
\begin{equation} \label{detcohomology}
 \lambda(E)_s := \det H^0(\mcr C_s, E_s) \otimes \det H^1(\mcr C_s, E_s)^{\vee},
\end{equation}
where ${\rm det}$ stands for the top exterior power of a vector space, and $(\cdot )^{\vee}$ for the dual. This bundle is called the {\em determinant of the cohomology} of $E \rightarrow \mcr{C}/S$. The fact that the complex lines (\ref{detcohomology}) glue together to form a line bundle $\lambda(E)$ over $S$ is the Knudsen--Mumford determinant construction (see~\cite{Knudsen-I}).

Whenever $E\rightarrow {\mcr C}$ carries a Hermitian metric,  then {\em Quillen's metric}  on \makebox{$\lambda(E)\rightarrow S$} is defined as follows.
  First of all, for every fibre $\mcr C_s$ the groups $H^i(\mcr C_s, E_s)$ can, by Hodge theory, be represented by $E_s$-valued harmonic forms, and this induces an
 $L^2$-metric in $\lambda(E) \rightarrow S$ (as in our
 discussion in Section~\ref{sec:L2geom}) that we shall denote as $h_{L^2}|_s$.
Consider the spectral zeta-function defined by
\begin{equation}\label{zeta}
t\;\mapsto\; \zeta_{\Delta_{\bar\partial}}(t) := \sum_{\lambda \in {\rm Spec}^+(\Delta_{\bar\partial})} \frac{1}{\lambda^t},
\end{equation}
where the sum is over all positive eigenvalues $\lambda$ of the Kodaira--Laplace operator
$\Delta_{\overline \partial}=\bar \partial \bar\partial^* + \bar\partial^* \bar\partial$ (and the adjoints are taken with respect to $h_{L^2}|_s$), acting on $\bar{E}_s$-valued smooth functions on the fibre $\mcr{C}_s$. For $\operatorname{Re}(t) \gg 0$,  the sum in (\ref{zeta}) converges absolutely, and it defines a holomorphic function which can be analytically continued to $t=0$. Then Quillen's metric $h_Q$ on $\lambda(E)$ is defined fiberwise as
\begin{equation} \label{Quillen}
h_Q|_s :=  \exp\left(\zeta'_{\Delta_{\bar\partial}}(0)\right)\, h_{L^2}|_s.
\end{equation}
We use the notation $\lambda(\bar E)_Q$ when we need to emphasise that 
we equip $\lambda(E)$ with the Hermitian metric (\ref{Quillen}), rather than the $L^2$-metric implied by writing $\overline{\lambda(E)}$.

\subsection{Deligne's pairing on families of curves and the norm functor}

We now introduce two constructions that will appear in computations of Quillen's metric via the Riemann--Roch formula for the curvature.

Given two line bundles $L$ and $M$ on $\mcr C$ for a family of curves $\mcr C/S$ over a complex manifold $S$, we recall that there is a natural line bundle on $S$, namely their {\em Deligne pairing}
\begin{equation}\label{Deligne}
\langle L, M \rangle := \lambda((L-\mathcal{O}_\mcr{C})\otimes (M-\mathcal{O}_{\mcr{C}})) := \lambda(L \otimes M) \otimes \lambda(L)^{-1} \otimes \lambda(M)^{-1} \otimes \lambda(\mathcal{O}_{\mcr C}).
\end{equation}
Moreover, given two smooth Hermitian metrics in $L$ and $M$, the line bundle
\makebox{$\langle L, M \rangle \rightarrow S$} carries a natural metric, induced by Quillen's  metrics on the different terms in (\ref{Deligne}). More concretely, we have (see also \cite[section 6]{determinant}, and \cite{Elkik} for higher-dimensional analogues):

\begin{definition} \label{Delignities}
Let the notation be as above.
\begin{enumerate}
\item The line bundle $\langle L, M \rangle \rightarrow S$ is defined in terms of generators and relations as follows:
\begin{itemize}
\item Whenever $\ell$ and $m$  are meromorphic sections of $L$ and $M$ such that $\ddiv{\ell}$ and $\ddiv{m}$ are relative divisors whose supports have empty intersection, we have a non-zero section
$\langle \ell, m  \rangle$ of $\langle L, M \rangle.$
\item If $\ell'$ is another rational section of $L$, then $f := \ell/\ell' \in \CC({\mcr C})$ satisfies 
$$\langle \ell, m \rangle (s)= f((\op{div} m)|_{{\mcr C}_s})  \langle \ell', m \rangle(s)\quad \text{ for }s\in S,$$
where $f(\sum_i n_i {\bf x}_i) = \prod_i f({\bf x}_i)^{n_i}$; and the corresponding statement on the second factor also holds.
\end{itemize}
\item If $L$ and $M$ carry Hermitian metrics, we define a Hermitian metric on $\langle L, M \rangle \rightarrow S$ (referred to as the natural metric) fibrewise by the formula
\begin{eqnarray*}
\log|\langle \ell, m \rangle|^2 & := & 
 \log|\ell|^2 ({\op{div} m)} + \log|m|^2( {\op{div} \ell)}
 + \frac{1}{\pi} \int_{{\mcr C}/S} {\partial}{\bar \partial} \left( -\log|\ell|\right) ^2 \log|m|^2
 \\
& = & \log|m|^2 ( {\op{div} \ell) } +  \int_{{\mcr C}/S} c(\overline{L}) \log |m|^2 .
\end{eqnarray*}
\end{enumerate}

\end{definition}

Let us make some comments to make the notation more explicit.
The equations above should
be interpreted fibrewise: by $\int_{{\mcr C}/S}$ we mean fibre integration, that is, the integral $\int_{{\mcr C}_s}$ for each $s\in S$; this is
precisely what is meant by the integral $\int_\Sigma$ in the formula (\ref{L2metricformula}).
Besides (and more generally), we will always write
$$
c(\overline{L}):= \frac{i}{2\pi}\, F_{A^{\rm Chern}}\; \in \;
\Omega^{1,1}(X;\RR)
$$
for the {\em first} Chern form of a Hermitian holomorphic line bundle $\overline L \rightarrow X$ on a complex manifold $X$, where $F_{A^{\rm Chern}} $ denotes the curvature of its Chern connection. Thus $[c(\overline{L})]=c_1(L) \in H^2(X;\ZZ)$ is the
first Chern class of $L\rightarrow X$ (an integral 2-cohomology class
defined independently of the geometry on the bundle).

We remark that the isomorphism $\langle L, M \rangle \simeq \langle M, L \rangle $ given by $\langle \ell, m \rangle \mapsto  \langle m, \ell \rangle$ is an isometry whenever $M$ and $L$ have Hermitian metrics, even though this might not be obvious from the formulas above.

 In this article, we will make use of the following two properties of the natural Hermitian metric on Deligne's pairings:
\begin{proposition}\label{deligneproductcurvature}
 Let $\overline{L}, \overline{M}$ be two Hermitian line bundles on $\mcr C$. If a family of curves ${\mcr C}/S$ is given, then we have that:
\begin{enumerate}
 \item  The following formula for the curvature holds: $$c(\overline{\langle L , M \rangle}) = \int_{\mcr{C}/S} c(\overline L) \wedge c(\overline M).$$
 \item  If $ L \simeq  L'$ is an isomorphism of holomorphic line bundles over $\mcr C$ endowed with Hermitian metrics $h$ and $h'$, then the squared norm of the induced isomorphism $$\langle \overline L, \overline M \rangle \simeq \langle \overline L', \overline M \rangle $$ is fibrewise given by the function $a: S\rightarrow \RR_{>0}$ with
 \begin{equation}\label{changemetricdeligne}
 a(s) := \exp \int_{{\mcr C}_s} \log (h|_s/h'|_s)\, c_1(\overline M).
 \end{equation}
\end{enumerate}
\end{proposition}

The second construction we will need is the {\em norm}.
For a relative divisor $D$ in ${\mcr C}/S$, this is a multiplicative functor 
\begin{equation}
N_{D/S}: \Picu(D) \longrightarrow \Picu(S).
\end{equation}
Here, $\Picu(X)$ denotes the category of line bundles on the variety $X$, whose objects are line bundles while its morphisms are isomorphisms. 
The definition is straightforward whenever the restriction $p:=\pi|_D: D \to S$ is a topological cover, whose definition we recall for the convenience of the reader: for a small enough open neighbourhood $U\subset S$ of any point, $p^{-1}(U)= \coprod_j D_j$ and there exist holomorphic maps $s_j:U \xrightarrow{\simeq} D_j$ with $p\circ s_j = {\rm id}_U$.
Now for each line bundle $\mcr{L}\in \Picu(D)$ we construct $N_{D/S}(\mcr{L})  \in \Picu(S) $ by gluing together local trivialisations over such trivialising open sets $U$:
\begin{equation} \label{def:norm}
N_{D/S}(\mcr{L})|_{U}:=\bigotimes_j s_j^*(\mcr{L}).
\end{equation}
If ${\mcr L} \rightarrow D$ comes endowed with a Hermitian metric, \eqref{def:norm} produces an induced Hermitian metric on $N_{D/S}(\mcr{L}) \rightarrow S$.

The two constructions introduced in this section are related as follows: for a relative divisor $D$ in ${\mcr C}/S$, there is a canonical isomorphism \begin{equation} \label{isonorm}
\langle \mathcal{O}(D), \mcr{L}\rangle \simeq N_{D/S}(\mcr{L}|_D).
\end{equation}
If $s_D$ denotes the canonical section of $\calo(D)$ with zero locus $D$, the isomorphism is defined by the map $\langle s_D, \ell \rangle \mapsto N_{D/S}(\ell|_D)$. If ${\calo(D)}$ and ${\mcr L}$ are equipped with Hermitian metrics, then the norm of the isomorphism in (\ref{isonorm}) is given by
\begin{equation}\label{normrestriction}
\exp\left(-\frac{1}{2}\int_{\mcr C/S} \log|s_D|^2 c(\overline{\mcr L})\right).
\end{equation}

\subsection{Riemann--Roch isomorphism and the curvature formula}

The Riemann--Roch isomorphism and the curvature formula relate the bundles and metrics introduced in the previous two sections. To be
more precise,  the following result expresses a nontrivial relationship between determinants of the cohomology and Deligne pairings on families of curves.
\begin{theorem} [\cite{determinant}, Th\'eor\`eme 9.9]  \label{Delignethm}
Let $\mcr{L} \in \Pic(\mcr{C})$. There are canonical isomorphisms
\begin{equation}\label{thmRR1}
\lambda(\mcr L)^{2} \otimes \lambda(\calo_{\mcr C})^{-2} \simeq \langle \mcr L, \mcr L \otimes K_{\mcr C/S}^{-1} \rangle\end{equation}
and 
\begin{equation}\label{thmRR2}
\lambda(\mcr L)^{12}  \simeq \langle K_{\mcr C/S}, K_{\mcr C/S} \rangle \otimes \langle \mcr L, \mcr L \otimes K_{\mcr C/S}^{-1} \rangle^6.\end{equation}
\end{theorem}
Deligne also proves that when these line bundles are endowed with Quillen's metrics and the natural metrics on Deligne's pairings in item (2) of Definition~\ref{Delignities}, \eqref{thmRR1} is an isometry, whereas \eqref{thmRR2} is an isometry up to a topological constant that only depends on the genus of any fiber of the family.
This implies the corollary that we shall state next, and which can be found in \cite{BisFree1, BisFree2}, though it is in fact used in the original proof of Theorem~\ref{Delignethm}.

For the convenience of the reader, we make the following recollection. For a K\"ahler (1,1)-form $\omega_M$ on a complex $n$-manifold $(M,j_M)$, one writes $g_M(\cdot,\cdot)=\omega_M(\cdot,j_M \cdot)$ for the underlying Riemannian metric. We denote its Ricci curvature by $\op{Ric}(\omega_M)$;
this is the (1,1)-form given by $\pp \log \det \left[ g_M\left( \frac{\partial}{\partial z_j}, \frac{\partial}{\partial \bar z_i}\right) \right]_{i,j=1}^n$ in terms of any local chart of holomorphic coordinates $z_1, \ldots, z_n$ of $M$. A basic fact in
K\"ahler geometry is that this form always represents the
first Chern class of (the tangent bundle of) $M$, or of the anti-canonical bundle $K^{-1}_M$, irrespective of the K\"ahler structure taken.

\begin{corollary} \label{DRRcurvature} Suppose $S$ is K\"ahler.
Let $\LL$ be a Hermitian line bundle on  the family of curves $\mcr{C}:=S\times \Sigma \rightarrow S$,  and $\omega_\Sigma$ a
K\"ahler form on $\Sigma$. Then
\begin{equation} \label{niceCor}
c\left(\lambda(\LL)_Q\right) = \frac{1}{2}\int_{\mcr{C}/S} c(\LL) \wedge (c(\LL) + \op{Ric} (\omega_\Sigma)).
\end{equation}
\end{corollary}
\noindent
Here we are using that $\lambda(\calo_{\mcr C})\rightarrow S$ is a trivial line bundle with constant metric, which has zero curvature. We have implicitly endowed (the tangent bundle of) $S \times \Sigma$ with a product K\"ahler metric for the construction of Quillen's metric
featuring in formula (\ref{niceCor}).

\section{The K\"ahler quantizations of 
${\mathcal{M}_d}$} \label{sec:qdata}

In this section we determine all possible K\"ahler quantizations of the K\"ahler manifold (\ref{nicerMd}), where $\omega_{L^2}$ is obtained from an area form $\omega_\Sigma$ assumed to satisfy the Weil integrality condition (\ref{rescaledWeil}). Recall from Section~\ref{sec:univ} that we denote by $\PD$ the line bundle corresponding to the (relative) universal degree-$d$ divisor ${\mcr D}_d$ on a family of curves $$\pi: S \times \Sigma \to S$$
determined by the first projection, 
where we take $S=S^d \Sigma\cong {\mathcal{M}_d}$ equipped with the induced complex structure $J^{j_\Sigma}$. In what follows, we shall consider the constructions reviewed in Section~\ref{sec:Quil} as applied to this $\pi$
while keeping the degree $d$ fixed. 

\subsection{Relation between the Picard groups of $\Sigma$ and $S^d\Sigma$}

Let $q \colon \Sigma^d \to S^d \Sigma$ be the quotient map, $p_i:\Sigma^d  \to \Sigma$ the $i$-th projection map, and $\tilde p:\Sigma^d \times \Sigma \rightarrow \Sigma$, 
$p:S^d\Sigma\times \Sigma \rightarrow \Sigma$
the projections onto the second factor. For a line bundle $L$ on $\Sigma$, we can consider the line bundle $L^{\boxtimes d} := \bigotimes_{i=1}^d p_i^\ast L \rightarrow \Sigma^d$. Since this bundle is $\mathfrak{S}_d$-invariant, it naturally descends to $S^d\Sigma$; we shall still denote this quotient by $L^{\boxtimes d}$, as no confusion will arise when the base is specified.

\begin{proposition} \label{Picmap} 
 The map $L \mapsto L^{\boxtimes d}$ just defined also admits the
 description
\begin{equation}\label{mapPics}
\Pic(\Sigma) \to \Pic(S^d \Sigma), \; L \mapsto \langle \PD, p^\ast L \rangle.
\end{equation}
It induces an isomorphism 
$$\Picf(\Sigma) \stackrel{\cong}{\longrightarrow} \Picf(S^d \Sigma)$$
between moduli spaces of flat line bundles.

\end{proposition}
\begin{proof}
By construction, it is enough to verify that $q^* \langle \mcr P_d, p^\ast L \rangle = L^{\boxtimes d}\rightarrow \Sigma^d$. Let $D_i$ be the divisor in $\Sigma^d \times \Sigma $ determined by image of the section $\sigma_i: \Sigma^d \to \Sigma^d \times \Sigma$ (to the projection onto the first factor) given by $ \sigma_i:  ({\bf z}_1 , \ldots  ,{\bf z}_d) \mapsto ({\bf z}_1 , \ldots  ,{\bf z}_d, {\bf z}_i)$; then $q^*\PD = \calo(\sum_{i=1}^d  D_i) $. It follows from \eqref{def:norm} and \eqref{isonorm} that $q^* \langle \mcr P_d, p^\ast L \rangle \simeq \bigotimes_{i=1}^d \sigma_i^* \tilde p^\ast L$. The first part of the proposition follows, since $\tilde p \circ \sigma_i = p_i$. 

For the second statement, notice that any flat line bundle $L \rightarrow \Sigma$ can be realized as a Hermitian holomorphic line bundle $\overline{L}$ with flat Chern connection, and this induces a flat Chern connection on the pull-back $p^* L$. If we are also given a Hermitian metric on $\PD$, there is a natural metric on the holomorphic line bundle $\langle \PD, p^*L \rangle \rightarrow S^d\Sigma$ as explained in point (2) of Definition~\ref{Delignities}, whose Chern connection is flat by Proposition \ref{deligneproductcurvature}. Thus $L \mapsto \langle \PD, p^\ast L \rangle$ indeed restricts to a map $\Picf(\Sigma) \to \Picf(S^d \Sigma)$. Now fix a point ${\bf x}\in \Sigma$. We claim the map $\alpha \colon \Sigma \to S^d \Sigma$ given by ${\bf z} \mapsto q({\bf z}, {\bf x}, \ldots, {\bf x}) = {\bf z} + (d-1){\bf x}$ induces an isomorphism $\alpha^\ast: \Picf(S^d \Sigma) \to \Picf(\Sigma)$. Since $L = \alpha^\ast L^{\boxtimes d}$, it would then follow immediately that the map $L \mapsto L^{\boxtimes d}$ is an isomorphism. 

To justify the claim, notice that we have a commutative diagram 
$$\xymatrix{ \Sigma \ar[r]^{\alpha} \ar[rd]_{{\rm AJ}_1} & S^d \Sigma \ar[d]^{{\rm AJ}_d} \\ 
& {\rm Jac}(\Sigma) }$$
where the downward arrows represent the Abel--Jacobi maps given by \makebox{${\rm AJ}_1(\bf{z}) = \int^{\bf{z}}_{\bf{x}}$} and \makebox{${\rm AJ}_d(D) = \int^D_{d\bf{x}}$} modulo periods, respectively. It is a
well-known fact that the induced map \makebox{$\Picf({\rm Jac}(\Sigma)) \to \Picf(\Sigma)$} is an isomorphism, so it follows that the claim amounts to the statement that the induced map $\Picf({\rm Jac}(\Sigma)) \to \Picf(S^d \Sigma)$ is an isomorphism. To show this, it suffices to establish injectivity, since we are dealing with a map of Abelian varieties.
But this is part of the result expressed in Theorem 19.7 of reference \cite{Pol}.
\end{proof}

The significance of this rather general result in the context of our article can be stated as follows.

\begin{corollary}\label{Picmap1}
Any Hermitian line bundle $\overline{\mcr M} \rightarrow S^d \Sigma$  equipped with a unitary flat connection $\nabla$ is isomorphic to a Deligne pairing $\langle \OPD, p^\ast \overline{L}\rangle$ with the natural metric and its Chern connection, for some Hermitian holomorphic line bundle $\overline{L} \rightarrow \Sigma$. 
\end{corollary}

\begin{proof}

Consider more generally a holomorphic line bundle $L\rightarrow Y$ on a compact complex manifold, equipped with two Hermitian metrics $\| - \|_1$ and $\| - \|_2$, with associated Chern connections $\nabla_1, \nabla_2$. Suppose that they define gauge-equivalent Chern connections, i.e. $\nabla_1 = \nabla_2 + {\rm d}\eta$ for a smooth function $\eta:Y\rightarrow \RR$. If $\ell:U\rightarrow L$ is a local holomorphic frame for $L$, we can write $\nabla_i = {\rm d} - \partial \log \| \ell\|_i^2$ for $i=1,2$, and conclude that $\nabla_1 - \nabla_2 = \partial \varphi$ where $\| - \|_2=\| - \|_1 e^{-\varphi}$ for a smooth function $\varphi: Y \to \bb R$. Hence $\overline{\partial} \eta =0$, so $\eta$ is holomorphic and necessarily constant since $Y$ is compact. It follows that the underlying Hermitian metric of the Chern connection is unique up to a multiplicative constant. 

The particular Hermitian metric on $\overline{\mcr M}$ in the statement of the corollary is thus unique up to a constant. By Proposition~\ref{Picmap} we know that $(\mcr{M}, \nabla)$ is of the shape $(\langle \PD, p^\ast {L}\rangle, \nabla^{\rm Chern})$, where $\nabla^{\rm Chern}$ is the Chern connection for the metric on Deligne pairings induced by suitable metrics on $\PD$ and $L$. If we modify the metric on $L$ by a constant $e^{-C}$ then, according to \eqref{changemetricdeligne}, the metric on the pairings will change by $\exp(C \cdot d)$; so we also obtain all scalar multiples in this way, and we conclude the proof. \end{proof}

\subsection{General construction of the K\"ahler quantizations}

Now we can turn directly to our main goal. Recall that, by \eqref{L2metricformula},
$$ \frac{1}{2\pi}\omega_{L^2} = -\frac{1}{2} \int_{\Sigma} c(\OPD) \wedge \left(c(\OPD) - \frac{\tau}{2 \pi} \, \omega_\Sigma \right),$$ which we rewrite as $$-\frac{1}{2}\int_{\Sigma}\left\{  c(\OPD) \wedge \left(c(\OPD) + \op{Ric}(\omega_{\Sigma})\right) - c(\OPD) \wedge \left(\op{Ric}(\omega_\Sigma) + \frac{\tau}{2 \pi}\, \omega_\Sigma \right) \right\} .$$
Making use of Corollary \ref{DRRcurvature}, we obtain  \begin{equation}\label{expressioncurvaturL2} \frac{1}{2\pi}\omega_{L^2} = c(\lambda(\OPD)^{-1}) +  \int_{\Sigma} c(\OPD) \wedge \left(\frac{1}{2}\op{Ric}(\omega_\Sigma) + \frac{\tau}{4 \pi}\, \omega_\Sigma\right).\end{equation}
This last equation motivates the following definition.

 \begin{definition} \label{quantizationdata} 
  Suppose that a K\"ahler (area) form $\omega_\Sigma \in \Omega^{(1,1)}(\Sigma;\RR)$  is given such that (\ref{rescaledWeil}) holds for a fixed $\tau$ satisfying (\ref{Bradlow}).
  \begin{enumerate}
 \item[(i)]
 We will denote (as in Example~\ref{excurve})
 \begin{equation}\label{kdefined}
 k:= \frac{\tau}{4 \pi} {\rm Vol}(\Sigma)= \frac{\tau}{4\pi}\int_\Sigma \omega_\Sigma \;\; \in \;\NN.
 \end{equation}
  \item[(ii)]
  Let $\overline{Q}\rightarrow \Sigma$ be any prequantization of
  $\left(\Sigma,\frac{\tau}{2}\omega_{\Sigma}\right)$,
  with holomorphic structure induced by the K\"ahler
  polarisation of $j_\Sigma$ (as in Example~\ref{excurve} but with a different normalisation of the area form).
  Let also $K_\Sigma^{\pm1/2} \rightarrow \Sigma$ denote 
  the spin bundle of $j_\Sigma$ (and its dual, respectively) determined by
  a choice of metaplectic structure on $\Sigma$, together with the (spin) Chern connection induced by the K\"ahler metric $g_\Sigma$ of $\omega_\Sigma$. Then we write
  \begin{equation} \label{bundleM}
  \overline{M}:=Q\otimes K_\Sigma^{-1/2} \rightarrow \Sigma
  \end{equation}
  for the tensor product bundle equipped with the product Hermitian structure and the product connection.
 \end{enumerate}
 \end{definition}

We are now ready to state the main result of this section, which describes all the prequantizations of the K\"ahler form $\omega_{L^2}$  on $S^d\Sigma$, polarised by the complex structure $J^{j_\Sigma}$.
Our formulation will take as input geometric data on $\Sigma$ encapsuled by the line bundle (\ref{bundleM}). We take advantage
of the universal degree-$d$ bundle $\overline{{\mcr P}_d} \rightarrow S^d\Sigma\times \Sigma$, equipped with the Hermitian structure induced from the line bundle $\overline{L}\rightarrow \Sigma$ where our vortices live (see Section~\ref{sec:Poincline}), to capture the $L^2$-geometry as suggested by
the formula (\ref{expressioncurvaturL2}).

\begin{theorem} [K\"ahler quantization of vortices]
\label{thmquantization}
Under the assumptions in Definition~\ref{quantizationdata}, the  Hermitian holomorphic line bundle 
\begin{equation}\label{prequantumbundle}
\overline{\mcr L}_{\overline{M}} := \lambda(\OPD)^{-1} \otimes \langle \OPD, p^\ast \overline{M} \rangle \longrightarrow S^d\Sigma,
\end{equation}
constructed from a line bundle $\overline{M}\rightarrow \Sigma$  of the kind specified in Definition \ref{quantizationdata}(ii), 
provides a $J^{j_\Sigma}$-polarised prequantization $(\overline{\mcr L}, \nabla)$ of $(S^d\Sigma, \omega_{L^2})$.
Conversely, all K\"ahler quantizations of $(S^d\Sigma,J^{j_\Sigma}, \omega_{L^2})$ can be obtained in this way up to isomorphism of  Hermitian holomorphic line bundles.
\end{theorem}
\begin{proof}
The bundle (\ref{bundleM}) clearly has curvature
\begin{equation}\label{curvatureM}
\frac{1}{2}\left (\op{Ric}(\omega_\Sigma) + \frac{\tau }{2 \pi}\, \omega_\Sigma \right);
\end{equation}
thus, referring to equation (\ref{expressioncurvaturL2}), it follows that $\overline{\mcr L}_{\overline {M}} \rightarrow S^d\Sigma$ does indeed determine a prequantum bundle for $\frac{1}{4\pi^2}\omega_{L^2}$. If $\overline{\mcr L} \rightarrow S^d\Sigma$ denotes any other prequantum bundle to $\omega_{L^2}$, then $\mcr L_{M} \otimes \mcr L^{-1}$ is flat, and hence by Corollary \ref{Picmap1} it is of the form $\langle \OPD, p^\ast \overline{L} \rangle$ for a unique flat line bundle $L\in \Picf(\Sigma)$. But then $\mcr L = \mcr L_{M \otimes L}$, and $M \otimes L$ is also of the form $Q' \otimes K_\Sigma^{-1/2}$ for a prequantum bundle $Q'\rightarrow \Sigma$ of $\left(\Sigma, \frac{\tau}{2}\omega_\Sigma\right)$.
\end{proof}

Since $S^d\Sigma$ is compact when $\Sigma$ is taken compact, we automatically obtain:
\begin{corollary}
A quantum Hilbert space for the prequantization determined by
$\overline{M}\rightarrow \Sigma$  is provided by the finite-dimensional
vector space
\begin{equation} \label{qHilbert}
\mathcal{H}^{\overline{M}}_{P_{J^{j_\Sigma}}}:= H^0\left(S^d\Sigma, \overline{\mcr L}_{\overline{M}}\right),
\end{equation}
equipped with the restriction of the Hermitian metric on $\Gamma\left(S^d\Sigma, \overline{\mcr L}_{\overline{M}}\right)$ determined by the K\"ahler form
$\omega_{L^2}$ on $S^d\Sigma$ and the Hermitian structure on (\ref{prequantumbundle}).
\end{corollary}

\subsection{The case of K\"ahler--Einstein metrics on $\Sigma$}

We shall now consider in more detail the particular situation where $\omega_\Sigma$ determines a metric $g_\Sigma$  of constant scalar curvature. In this case, there is a natural quantization of $\mathcal{M}_d$ (for each $d$) which dispenses the input data from Definition \ref{quantizationdata}(ii). Indeed, suppose that the K\"ahler--Einstein equation \begin{equation}\label{oldKEequation}
\op{Ric}(\omega_\Sigma) = - \frac{\tau}{2 \pi} \, \omega_\Sigma
\end{equation}
is satisfied. Then by \eqref{expressioncurvaturL2} the curvature of Quillen's metric on $\lambda(\OPD)^{-1}$ is $\frac{1}{2\pi}\omega_{L^2}$. In this case, integrating \eqref{KEequation} over $\Sigma$ and applying the theorem of Gau\ss--Bonnet, we infer that $\tau \op{Vol}(\Sigma) = 2\pi (2g -2)$. By assumption \eqref{Brad} we must then have $g-1 \geq d > 0$; thus $\Sigma$ is necessarily hyperbolic, and equation \eqref{oldKEequation} can be rewritten as
\begin{equation} \label{KEequation}
\op{Ric}(\omega_\Sigma) = - \frac{2g-2}{{\rm Vol}(\Sigma)} \, \omega_\Sigma.
\end{equation}

Referring back to  Theorem~\ref{thmquantization}, we have established the following:  
\begin{corollary}\label{KE-canonicalquant} Suppose that $g-1>d$ and that the K\"ahler--Einstein equation (\ref{KEequation}) holds on $\Sigma$. Then the determinant of the cohomology with its Quillen metric,  $\lambda(\OPD)^{-1}\rightarrow S^d\Sigma$, canonically provides a  K\"ahler quantization of $(S^d\Sigma, \omega_{L^2})$.
\end{corollary}
In fact, the special geometry of the situation allows us to say a little more about this particular choice of metric. We phrase such statements as an informal remark.

\begin{remark}\label{remark:tangentbdle}
There is an exact sequence of line bundles on $S^d\Sigma \times \Sigma$
 $$0 \to \calo \to \mcr P_d \to \calo_{{\mcr D}_d}(\mcr D_d) \to 0,$$
induced by the tautological section $\Phi$ of $\PD = \calo({\mcr D}_d)$.
The determinant of cohomology for the family of curves $\pi \colon S^d\Sigma \times \Sigma \to S^d \Sigma$ is multiplicative on short exact sequences, and so we have an isomorphism
$$\lambda(\mcr P_d) \simeq \lambda(\calo) \otimes \det \pi_* \calo_{{\mcr D}_d}({\mcr D}_d).$$
It turns out that the direct image bundle
$\pi_* \calo_{{\mcr D}_d}({\mcr D}_d)$ is the tangent bundle of $S^d \Sigma$ (see \cite[IV, Lemma 2.3]{arabello}), and it is easy to see that the bundle $\lambda(\calo)$ is trivial. Since $\det {\rm T}{S^d \Sigma} = K_{S^d\Sigma}^{-1}$, we deduce that $\lambda(\mcr P_d)^{-1}$ is the canonical bundle on $S^d\Sigma$, whose first Chern class can be represented by $-\operatorname{Ric}(\omega)$ for any K\"ahler metric $\omega$ on $S^d\Sigma$.  In particular, whenever the K\"ahler--Einstein equation (\ref{KEequation}) is satisfied, we have an equality of cohomology classes $$\left[\operatorname{Ric}\left(\frac{1}{2\pi}\omega_{L^2}\right)\right] = -\left[\frac{1}{2\pi}\omega_{L^2}\right].$$
May it be true that this equation holds at the level of 2-forms? In other words, given the hyperbolic metric on $\Sigma$, is the associated $L^2$-metric on the vortex moduli space $S^d \Sigma$ also of constant scalar curvature? Notice that if $d=1$, then $S^1\Sigma = \Sigma$ and, by uniqueness of the solutions to the K\"ahler--Einstein equation, we would necessarily need  $\omega_{L^2}$ to be the hyperbolic metric itself. Though the corresponding statement does hold for one vortex on the hyperbolic disc~\cite{Str}, the extension to a compact hyperbolic surface is
nontrivial and would even provide a surprising twist to recent 
results on solutions of
the vortex equations in integrable situations (which do require (\ref{KEequation}) as an assumption~\cite{MalMan}, but yield no information about the metrics on the moduli space).
\end{remark}

\section{Holomorphic wavesections as multi-spinors on $\Sigma$} \label{sec:dims}

In this section we shall establish a quantum equivalent, within the geometric quantization framework developed in Section~\ref{sec:qdata}, of the picture (provided by Theorem~\ref{moduliSd}) of the moduli space $\mathcal{M}_d \cong S^d\Sigma$ as a classical phase space of indistinguishable point particles on the surface $\Sigma$. 

Again, we shall work under the assumptions stated in Definition~\ref{quantizationdata}, and comply with the conventions
for $k$, $\overline{K}_\Sigma^{1/2}\rightarrow \Sigma$ and $\overline{Q}\rightarrow \Sigma$ stated there. As before, the genus of $\Sigma$ will be denoted by $g$.
 
\subsection{Main result and its first consequences}

The main result of this section is the following theorem which, in addition to providing a
direct way of computing the dimensions of the quantum Hilbert spaces (\ref{qHilbert}), 
supplies a tie between our quantization scheme and the standard way of describing multiparticle states in elementary nonrelativistic quantum mechanics. 
 
\begin{theorem} \label{thmdims}
There is a natural isomorphism of spaces of holomorphic sections 
\begin{equation}\label{transferral}
H^0(S^d \Sigma, \mcr L_M) \cong \bigwedge{\!}^d\, H^{0}(\Sigma, Q \otimes K_{\Sigma}^{1/2}).
\end{equation}
Moreover, we always have $$\op{dim} H^{0}(\Sigma, Q \otimes K_{\Sigma}^{1/2}) \geq k,$$ with equality (leading to $\dim H^0(S^d\Sigma,{\mcr L}_M)= {k\choose d}$) whenever $k > g-1$.  
\end{theorem}

\begin{proof}  
We first remark that it follows from Lemma \ref{Picmap} and Remark \ref{remark:tangentbdle} that our prequantum bundle $\mcr L_M \rightarrow S^d\Sigma$ is of the form  $K_{S^d \Sigma} \otimes M^{\boxtimes d}$, for the line bundle $M\to  \Sigma$ in \eqref{bundleM}. By Serre duality,  $H^0(S^d \Sigma, K_{S^d \Sigma} \otimes M^{\boxtimes d}) \cong H^d(S^d \Sigma, (M^{-1})^{\boxtimes d})^\vee$. Now observe that, since the complex dimension of $\Sigma$ is 1, the groups $H^i(\Sigma, {\mcr{F}})$ with $i>1$ vanish for any sheaf $\mcr F$ on $\Sigma$. Then we infer from K\"unneth formula for sheaf cohomology~\cite[p.~114]{KemAV}  that the cup product induces an isomorphism $$ H^1(\Sigma, M)^{\otimes d} \stackrel{\cong}{\longrightarrow} H^d(\Sigma^d, M^{\boxtimes d}).$$
Note that the bundle $M^{\boxtimes d} \rightarrow \Sigma^d$ is $\mathfrak{S}_d$-equivariant. By a Hochschild--Serre spectral sequence argument, the cohomology group  $H^d(S^d \Sigma, (M^{-1})^{\boxtimes d})$  identifies with the $\mathfrak{S}_d$-invariant part of $H^d(\Sigma^d,(M^{-1})^{\boxtimes d})$. Since the cup product is anti-commutative, this $\mathfrak{S}_d$-invariant part is generated by the alternating vectors in $ H^1(\Sigma, M^{-1})^{\otimes d}$, which we can identify with $\bigwedge^d H^1(\Sigma, M^{-1})$. Hence, again by Serre duality, we obtain 
\begin{eqnarray*}
H^0\left(S^d \Sigma, K_{S^d \Sigma} \otimes  M^{\boxtimes d}\right) &\cong& \left(\bigwedge{\!\!}^d \,H^0(\Sigma, K_\Sigma \otimes M)^\vee\right)^\vee \\
&\cong& \bigwedge{\!\!}^d \, H^0(\Sigma, K_\Sigma \otimes M)  .
\end{eqnarray*}
Since $M = Q \otimes K_{\Sigma}^{-1/2}$, we infer that $$H^0(S^d \Sigma, K_{S^d \Sigma} \otimes  M^{\boxtimes d})  \cong \bigwedge^d H^0(\Sigma, Q \otimes K_\Sigma^{1/2}).$$ 
Our choice of isomorphism $(\bigwedge^d(V^\vee)^\vee \cong\bigwedge^d V$ is induced by the pairing given by linearly extending the assignment
$$(e_1 \wedge e_2 \wedge \ldots \wedge e_d) \otimes (f_1 \wedge f_2 \wedge \ldots \wedge f_d) \mapsto \sum_{\sigma \in \mathfrak{S}_d} (-1)^\sigma\prod_{i=1}^d f_i(e_{\sigma (i)}).$$

It follows from Definition~\ref{quantizationdata}  the degree of the bundle $Q \otimes K_\Sigma^{1/2}$ is $k + g-1$. By Riemann--Roch, we then have 
\begin{eqnarray*}
h^0(\Sigma, Q \otimes K_\Sigma^{1/2}) &\geq& h^0(\Sigma, Q \otimes K_\Sigma^{1/2}) - h^1(\Sigma, Q \otimes K_\Sigma^{1/2}) \\
&=& (k+g-1) + (1-g) = k.
\end{eqnarray*}
Whenever $k > g-1$, we have $h^1(\Sigma, Q \otimes K_\Sigma^{1/2}) = 0$, thus in that case we obtain an equality $$h^0(\Sigma, Q \otimes K_\Sigma^{1/2}) = k$$
from this Riemann--Roch argument. 
\end{proof}

\begin{remark}
Recall that $k > d$ by assumption (\ref{Bradlow}). Hence Theorem~\ref{thmdims} implies that there always exist nontrivial wavesections in our K\"ahler quantization problem. 
\end{remark}

\begin{remark}
The spaces of holomorphic sections on either side of (\ref{transferral}) come naturally equipped with Hermitian inner products (induced from $L^2$-metrics on sections and Quillen's construction). Our statement in Theorem~\ref{thmdims} ignores whether the isomorphism \eqref{transferral} is isometric --- and we do not even expect this to be the case. Note that the inner product that is physically significant is the one on the left-hand side, but expressing it in terms of a basis of
the right-hand side is a highly nontrivial task.
\end{remark}

The gist of Theorem~\ref{thmdims} is that there exists a canonical fashion to effectively interpret any wavesection
$$ \psi:S^d \Sigma \rightarrow {{\mcr L}_M}\qquad \text{ for any }M=M(Q,K_\Sigma^{1/2})$$
in our K\"ahler quantization scheme as the sum of  alternating products 
$\psi_1 \wedge\cdots \wedge \psi_d$
of $d$   holomorphic spinors $\psi_i: \Sigma \rightarrow  Q \otimes K_\Sigma^{1/2} $ on the original surface $\Sigma$, taking values in the prequantum bundle $Q\rightarrow \Sigma$ to $(\Sigma,\frac{\tau}{2}\omega_\Sigma)$ used as an ingredient to construct $M$. This result has two immediate consequences:
\begin{itemize}
\item[(1)] It signifies that the quantum multivortex states represented by our quantization scheme have {\em fermionic} character, since any wavesection $\psi$ is multiplied
by the sign $(-1)^{\sigma}$ of a permutation $\sigma \in \mathfrak{S}_d$ acting on the effective (and indistinguishable) one-particle states $\psi_i$
coming from each of the $d$ copies of $H^0(\Sigma,Q\otimes K^{1/2}_\Sigma)$.
\item[(2)] The $Q$-valued $j_\Sigma$-spinors $\psi_i$ are automatically
half-forms on $\Sigma$, so the alternative description of the quantization supplied by the isomorphism (\ref{transferral}) is a multiparticle
{\em half-form quantization} on the surface $\Sigma$. And this may sound somewhat surprising, since we did not start from a half-form quantization scheme
to construct the wavesections $\psi$.
\end{itemize}

Further to 
(2), the next result shows that half-form K\"ahler
quantization does not even apply to moduli spaces of multivortices, unless $g=0$ and $d$ is odd.

\begin{proposition}\label{nohalfforms}
Given a compact Riemann surface $\Sigma$ of genus $g$, the manifold $S^d\Sigma$ admits a metaplectic structure if and only if $d=1$, or if  $g=0$ and $d$ is odd. 
\end{proposition}
\begin{proof}
Recall that a manifold $X$ admits a spin structure if and only if the
second Stieffel--Whitney class $w_2({\rm T}X)\in H^2(X;\ZZ/2\ZZ)$ vanishes~\cite{BorHir}. Since $S^d\Sigma$ is a complex manifold,
one has that 
\begin{equation} \label{w2fromc1}
w_2({\rm T}S^d\Sigma) \equiv c_1({\rm T}S^d \Sigma)\,  ({\rm mod}\, 2);
\end{equation}
on the other hand, a result of \cite[p. 334]{McD} is that
\begin{equation} \label{c1SdSigma}
c_1({\rm T}S^d\Sigma) = (d+1-g) \eta - \theta \;\in \; H^2(S^d \Sigma ;\ZZ),
\end{equation}
with $\eta$ and $\theta$ as defined after equation (\ref{omegacoh}). 
The case $d=1$ is classical,  
and to deal with the $d > 1$ case, we start by observing that
the isomorphism (\ref{isoXi}) is also valid by changing coefficients to
$\ZTWO$, since all the groups involved are torsion free (see (12.3) in \cite{McD}). Then we read off from (\ref{c1SdSigma}) that the vanishing of \eqref{w2fromc1} is
equivalent to 
\begin{equation}\label{formercondition}
(d+1-g)\,  \eta \equiv 0 \, ({\rm mod}\, 2)\quad \text{ in }  H_0(\Sigma; \ZTWO),
\end{equation}
together with
\begin{equation} \label{lattercondition}
\theta \equiv 0 \, ({\rm mod}\, 2) \quad \text{ in } \bigwedge{\!}^2\, H_1(\Sigma; \ZTWO).
\end{equation} 
The condition (\ref{formercondition}) is equivalent to $d+1-g$ being even, and we claim that this cannot be true together with (\ref{lattercondition}), unless $g=0$. 

To justify our claim, we first observe that if $\theta=2\theta'$ for some other class \makebox{$\theta'\in H^2(S^d\Sigma;\ZZ)$,} then the corresponding statement would also be true in $S^{d'} \Sigma$ for any $d'>1$; this is
because (\ref{isoXi}) provides an intrinsic description of $H^2(S^d\Sigma;\ZZ)$ in terms of the homology of $\Sigma$, which is independent of $d$. Thus we can complete the argument by reducing it to a fixed $d$, which we
set to be $d=g$ for convenience. Now if we insert $\theta=2\theta'$ in the relation (\ref{thetag}),
we obtain the implication $2^g \, |\, g!$. By Legendre's formula in number theory~\cite[p.~77]{Mol}, this assertion only holds for $g=0$; and if $g=0$, we have $\theta=0$ by definition.
\end{proof}

In addition to the basic consequences (1) and (2) above, the fact that they occur {\em simultaneously} is significant from the point of view of quantum mechanics. It can be interpreted as a manifestation of the ``spin-statistics theorem'' --- a general principle postulating 
that particles associated to fundamental spinors should be assigned fermionic statistics. Rather than imposing it as an independent axiom, in quantum theory one strives for a derivation of this principle as a consistency condition on more foundational assumptions, and our results for quantization of vortices are in this spirit.

\begin{remark}
A fermionic
character to quantized vortices has been inferred also in a semiclassical
approach to the canonical quantisation of an A-twisted version of a supersymmetric extension of the Abelian Higgs model on compact surfaces (see \cite{BokRom} for a discussion). The relation of that work to the interpretation in item (1) above remains obscure. 
\end{remark}

\begin{remark}
The case $g=0$, already considered in \cite{RomQVS}, is very particular, for then one can interpret K\"ahler quantization alternatively as yielding {\em bosonic} quantum effective particles on $\Sigma \cong \PP^1$.
This is essentially because $S^d(\PP^1)\cong \PP^d$, for which
${\rm Pic}(\PP^d)$ is cyclic with generator $\mathcal{O}_{\PP^d}(1)$.
The prequantum line bundle to a K\"ahler form with K\"ahler class
(\ref{omegacoh}) will necessarily be of the form $\mathcal{O}_{\PP^d}(\ell):=\mathcal{O}_{\PP^d}(1)^{\otimes \ell}$
for $\ell= k - d \in \NN$, and
furthermore $H^0(\PP^d,\mathcal{O}_{\PP^d}(1)) \cong S^d(H^0(\PP^1,\mathcal{O}_{\PP^1}(1)))$. This leads to a description
of the quantum Hilbert space in K\"ahler quantization 
$$
H^0(\PP^d,{\mathcal O}_{\PP^d}(\ell)) \cong S^d \left( H^0 (\PP^1, \mathcal{O}_{\PP^1}(\ell))\right).
$$
The alternative between a bosonic and a fermionic interpretation
was already pointed out in \cite[p.~3465]{RomQVS} (where the fermionic alternative required a particular level in Chern--Simons theory, and was disfavoured for this reason). We should emphasize here that the bosonic alternative is no longer avalilable for $g>0$.
\end{remark}

\subsection{Special divisors and change of complex structure}
\label{sec:specialdiv}

Even though the assumption $k>g-1$ simplifies (and stabilises) the calculation  of the dimension of the Hilbert spaces (\ref{qHilbert}), as stated in Theorem~\ref{thmdims}, it is instructive to consider the effect of \begin{equation} \label{specialdiv}
h^1(\Sigma,Q\otimes K_\Sigma^{1/2}) \ne 0,
\end{equation}
i.e. the situation where our $Q$-valued spinors $\psi_i: \Sigma \rightarrow Q  \otimes K_{\Sigma}^{1/2}$ define special divisors on $\Sigma$. This
is the subject of Brill--Noether theory in the classical geometry of curves; see \cite{arabello}. 

Recall that (\ref{Bradlow}) forces $k$ in (\ref{kdefined}) to satisfy
\begin{equation} \label{degreeg}
{\rm deg}(Q\otimes K_\Sigma^{1/2})=k+g-1>g,
\end{equation}
which means that the condition (\ref{specialdiv}) is non-generic among line bundles of degree $k+g-1$ on $\Sigma$. Moreover,
it can only hold for degrees $k + g -1 \le 2g-2$. i.e. for metrics on
$\Sigma$ whose total area is sufficiently small:
\begin{equation} \label{smallenoughV}
{\rm Vol}(\Sigma) \le 4 \pi(g-1).
\end{equation}

The nonvanishing (\ref{specialdiv}) leads to an enhancement of the dimension of the quantum Hilbert space to
$$
\dim \mathcal{H}^{\overline{M}}_{P_{J^{j_\Sigma}}}=
{k + h^1(\Sigma, Q \otimes K_\Sigma^{1/2}) \choose d} > {k \choose d},
$$ 
where the right-hand side corresponds to the generic dimension. 
The simplest example of jumping dimensions occurs when the degree (\ref{degreeg}) is $2g-2$, i.e. the total area in (\ref{smallenoughV}) is
precisely $4\pi(g-1)$. Then there is exactly one line bundle
in ${\rm Pic}^{2-2g}(\Sigma)$ satisfying (\ref{specialdiv}), namely the
canonical bundle (which is attained when $Q=K_\Sigma^{1/2}$), on
which $h^1(\Sigma,Q\otimes K_\Sigma^{1/2})=1$. For lower degrees, i.e.
$1<k<g-2$, the situation is more complicated, and in general the pattern of jumping dimensions is even sensitive to the complex structure $j_\Sigma$
chosen, for a fixed genus $g$.

The occurrence of jumps in the dimensions of the quantum Hilbert spaces heralds the fact that  different choices of $\overline{M}\rightarrow \Sigma$, as in Definition~\ref{quantizationdata}(ii), cannot just correspond to different representations (in the sense of quantum mechanics) of the same quantum system, since the
Hilbert spaces are not isomorphic in general --- not even through non-unitary isomorphisms. At the very best, one could ask whether quantum Hilbert spaces $\mathcal{H}^{\overline{M}}_{P_{J^{j_\Sigma}}}$
corresponding to $M\rightarrow \Sigma$ within different strata in ${\rm Pic}^{k+g-1}(\Sigma)$ (according to the value of $h^1$) could be related through unitary isomorphism. We shall return to this question in the next section.

The fact that the pattern of jumps  may depend on $j_\Sigma$
shows {\em a posteriori\,} that, in the quantization of the moduli space of vortices, the variation of the preferred polarisation on $S^d\Sigma$ (i.e. of the complex structure $J^{j_\Sigma}$ determined by $j_\Sigma$) has a very different flavour to the one that is familiar from K\"ahler quantization of a K\"ahler manifold with symplectic structure independent
of the choice of adapted complex structure.
The K\"ahler quantization of the family of K\"ahler manifolds (\ref{nicerMd}) obtained by variation of $j_\Sigma$ should be treated as separate quantization problems altogether, in which there is a preferred complex structure induced by $j_\Sigma$ that is part of the classical data.

 \section{On relating choices via projectively flat connections} \label{sec:projflat}

For a vector bundle $E$ with connection $\nabla$, we say that $\nabla$ is {\em projectively flat} if its curvature $F_\nabla$ is a 2-form valued in the centre of ${\rm End}(E)$.  In this situation, parallel transport of vectors allows us to identify the fibres of $\bb P(E)$.
In geometric quantization of the moduli space of flat connections, a projectively flat connection over the space parametrising complex polarizations is a useful device  to identify projectivisations of quantum Hilbert spaces
corresponding to different polarizations~\cite{Hit};  for an account of the latest refurbishment of this technique (which goes under the name of {\em Hitchin's connection}), we refer the reader to~\cite{AndRas}.

Our considerations in Section~\ref{sec:dims} have shown that, once a prequantization of $(\Sigma,\frac{\tau}{2}{\omega_\Sigma})$ together with a metaplectic correction, are fixed, we can model the quantum
Hilbert space in
the K\"ahler quantization of (\ref{nicerMd}) by the vector space $\bigwedge^d H^0(\Sigma,Q\otimes K^{1/2}_\Sigma)$. All choices involved are obtained by fixing one single
metaplectic structure on the surface $\Sigma$, and varying the  prequantum bundle $\overline{Q} \in {\rm Pic}^k(\Sigma)$ of $(\Sigma, \frac{\tau}{4\pi} \omega_{\Sigma})$. In Section~\ref{sec:specialdiv} we emphasised that, in general, there may well be jumps in the dimensions of the quantum
Hilbert spaces, which are therefore manifestly not isomorphic. But can one at least identify projectively the Hilbert spaces within strata where the dimensions are kept constant? In what follows,
we want to look at this problem in the simplest situation where there is a single stratum, by assuming that $k > g-1$. Then by Theorem \ref{thmdims}  all the K\"ahler quantizations in our scheme 
glue together to form a vector bundle of rank $k\choose d$ on ${\rm Pic}^{k+g-1}(\Sigma)$, and one may ask whether one can construct a projectively flat connection on this bundle.

To be more precise, let ${\mathcal{P}}_{k+g-1}$ be any Poincar\'e line bundle (of degree $k+g-1$) on the family of curves ${\rm pr}_1 \colon \operatorname{Pic}^{k+g-1}(\Sigma) \times \Sigma \to \operatorname{Pic}^{k+g-1}(\Sigma)$ (see \cite[p.~166]{arabello}).  Then ${\rm pr}_{1 \ast} \mathcal P_{k+g-1} \rightarrow  {\rm Pic}^{k+g-1}(\Sigma)$ is a vector bundle whose fibres are of the form $H^0(Q \otimes K_{\Sigma}^{1/2})$, by Grauert's theorem on direct images of coherent sheaves.  The conclusion is that, for $k>g-1$, there is a vector bundle 
\begin{equation} \label{bundleonjac}
\bigwedge{\!\!}^d\;{\rm pr}_{1*} \mathcal P_{k+g-1} \longrightarrow {\rm Pic}^{k+g-1}(\Sigma)
\end{equation}
whose base parametrises the different choices of quantization data, and whose fibres model the corresponding quantum Hilbert spaces (\ref{qHilbert}). 

In the following theorem we answer negatively, using topological obstructions, the question of existence of a projective flat connection in the vector bundle \eqref{bundleonjac} in nontrivial cases for genus $g > 1$. Note that this issue only arises for $g>0$, since the Picard variety of $\Sigma=\mathbb{P}^1$ in a given degree is a point.

\begin{theorem} \label{noHit}
Suppose that $g>1$, $k> g-1$ and $k \geq d > 0$. Then there exists a projectively flat connection in the vector bundle (\ref{bundleonjac}) if and only if we are in the degenerate situation of dissolved vortices, i.e. when $k = d$.  
\end{theorem}
\begin{proof}  Suppose $\tilde E$ is any vector bundle of rank $\tilde r$ admitting a connection
which is projectively flat. Then, using the definition of Chern forms in terms of the curvature matrix, we deduce the relationship of Chern classes
\begin{equation} \label{necessary}
(\tilde r -1)c_1(\tilde E)^2 = 2\tilde r \, c_2(\tilde E).
\end{equation}

The Chern classes of $E := {\rm pr}_{1*}\mathcal{P}_{k+g-1}$ have been computed by Mattuck \cite[Theorem 4]{Mattuck}. For the statement, fix a point in $ \Sigma$ and consider the associated Abel embedding $\Sigma \hookrightarrow \Jac(\Sigma)$.
We denote by $W_{i}$  the Poincar\'e dual to the image of $S^{g-i}\Sigma \to {\rm Pic}^{g-i}(\Sigma) \to {\rm Jac}(\Sigma)$, where we make use of the additive structure on the Jacobian. Then it is  shown in \cite{Mattuck} that we have equalities
$$c_1(E) = -W_1,\quad c_2(E) = W_1^2-W_2$$
in the Chow ring, hence also in de Rham cohomology. Moreover, if 
$\Theta \in H^{2g}({\rm Jac}(\Sigma);\ZZ)$ denotes the theta-class as before, in de Rham cohomology a formula of Poincar\'e (see \cite{GH}, p. 350)
states that  $W_2 = \frac{1}{2!}\Theta^2$. Since $\Theta$ is defined as $W_1$, we deduce that $c_1(E)=-\Theta $ and $c_2(E) =  \frac{1}{2}\Theta^2$.

It is not hard to show, using expressions for Chern classes in terms of Chern roots, that for a general vector bundle $E$ of rank $r$ we have $$c_1(\bigwedge{\!}^d\, E) = {r-1 \choose d-1} c_1(E) $$
and
$$c_2(\bigwedge{\!}^d\, E) = {r-2 \choose d-1} c_2(E) + \frac{ {r -1 \choose d-1 }^2 -  {r -1 \choose d-1 } }{2}c_1(E)^2.$$
Thus the necessary relation (\ref{necessary}) for the bundle $\tilde E=\bigwedge{\!}^d\, E$ of rank $\tilde r={ r \choose d}$, where $r={\rm rk}\, ({\rm pr}_{1*} \mathcal{P}_{k+g-1}) =k$, becomes
\begin{equation}\label{messy}
\left( {k \choose d } -1 \right) {k- 1 \choose d-1}^2 \Theta^2 = {k \choose d} \left( {k-1 \choose d-1}^2 - \frac{d-1}{k-1} {k-1  \choose d-1} \right) \Theta^2.
\end{equation}
This condition is satisfied for $k=d$. Suppose now that $k>d$. Since we are assuming that $g>1$, and $\Theta^g \ne 0$ (see the discussion preceding \eqref{thetag}), $\Theta^2$ cannot be a torsion element, so
(\ref{messy}) simplifies to
$$ {k- 1 \choose d-1}  = \frac{d-1}{k-1}  {k \choose d}.$$
But this is equivalent to $k=d$, which we had already considered.

Hence  $\tilde E = \bigwedge^d E$ satisfies condition (\ref{necessary}) if and only if  $k=d$, in which case this bundle is a line bundle,
and so  any connection in it is (trivially) projectively flat.
\end{proof}

Dissolved vortices (corresponding to the critical value $\tau=\frac{4 \pi d}{{\rm Vol}(\Sigma)}$, see \eqref{Brad}) were discussed in \cite{ManRom}. In this degenerate situation, (\ref{L2norm}) implies that $\phi=0$ and the moduli space, with its $L^2$-geometry, reduces to Jacobian of $\Sigma$ with its usual flat geometry --- irrespective of the degree $d$ --- as suggested by the formula (\ref{omegacoh}). Geometric quantization in this
context (see e.g.\ \cite{Mou_etal}) has a different flavour to the discussion in this paper.

\bibliographystyle{numsty}

\begin{thebibliography}{9999}
                        \setlength{\itemsep}{0 pt}
                        \setlength{\parskip}{0pt}
                        \setlength{\parsep}{0pt}

\begin{small}

\bibitem{AndGamLau}
\textsc{J.E. Andersen, N.L. Gammelgaard {\upshape and} M.R. Lauridsen}: Hitchin's connection in metaplectic quantization.
{\sl Quant. Top.} {\bf 3} (2012) 327--357 

\bibitem{AndRas}
\textsc{J.E. Andersen {\upshape and} K. Rasmussen}: A Hitchin connection for a large class of families of K\"ahler structures;
{\tt arXiv:1609.01395}

\bibitem{arabello}
\textsc{E. Arbarello, M. Cornalba, P.A. Griffiths {\upshape and} J. Harris}:  {\it Geometry of Algebraic Curves}, Volume I, Springer-Verlag, 1985

\bibitem{Ati}
\textsc{M.F. Atiyah}: Topological quantum field theory.
{\sl Publ. Math. IH\'ES}\, {\bf 68} (1988) 175--186 

\bibitem{AxdPiWi}
\textsc{S. Axelrod, S. della Pietra {\upshape and} E. Witten}: Geometric quantization of Chern--Simons gauge theory.
{\sl J. Diff. Geom.} {\bf 33} (1991) 787--902 


\bibitem{Mou_etal}
\textsc{T. Baier, J.M. Mour\~ao {\upshape and} J.P. Nunes}: Quantization
of Abelian varieties: distributional sections and the transition from
K\"ahler to real polarizations.
{\sl J. Funct. Anal.} {\bf 258} (2010) 3388--3412 


\bibitem{BapL2M}
\textsc{J.M. Baptista}: On the $L^2$-metric of vortex moduli spaces.
\newblock \textsl{Nucl.\ Phys.\ B} \textbf{844} (2011) 308--333


\bibitem{BatWeiGQ}
\textsc{S. Bates {\upshape and} A. Weinstein}: {\it Lectures on the Geometry of Quantization},
American Mathematical Society, 1997

\bibitem{BerTha}
\textsc{A. Bertram {\upshape and} M. Thaddeus}:
On the quantum cohomology of a symmetric product of an algebraic
curve.
\newblock \textsl{Duke J. Math.} \textbf{108} (2001) 329--362. 

\bibitem{BisFree1}
\textsc{J.-M. Bismut {\upshape and} D. Freed}:
The analysis of elliptic families: I. Metrics and connections on determinant bundles. 
\newblock \textsl{Comm. Math. Phys.} \textbf{106} (1986) 159--176. 

\bibitem{BisFree2}
\textsc{J.-M. Bismut {\upshape and} D. Freed}:
The analysis of elliptic families II: Dirac operators, eta invariants, and the holonomy theorem. 
\newblock \textsl{Comm. Math. Phys.} \textbf{107} (1986) 103--163. 

\bibitem{BisRag}
\textsc{I. Biswas {\upshape and} N. Raghavendra}: The determinant bundle on the moduli space of stable triples over a curve.
\newblock \textsl{Proc. Indian Acad. Sci. (Math. Sci.)} \textbf{112} (2002) 367--382


\bibitem{BisSch}
\textsc{I. Biswas {\upshape and} G. Schumacher}: Coupled vortex equations and moduli: deformation theoretic approach and K\"ahler geometry.
\newblock \textsl{Math. Ann.} \textbf{343} (2009) 825--851


\bibitem{Bog}
\textsc{E.B. Bogomol'ny\u\i}: The stability of classical solutions.
\newblock  \textsl{Sov. J. Nucl. Phys.} \textbf{24} (1976) 449--454

\bibitem{BokRom}
\textsc{M. B\"okstedt {\upshape and} N.M. Rom\~ao}: Pairs of pants, Pochhammer curves and $L^2$-invariants;
{\tt arXiv:1410.2429}

\bibitem{BorHir}
\textsc{A. Borel {\upshape and} F. Hirzebruch}: 
Characteristic classes and homogeneous spaces I.
\newblock \textsl{Amer.\ J.\ Math.} \textbf{80} (1958) 97--136



\bibitem{BraVHLB}
\textsc{S.B. Bradlow}: Vortices in holomorphic line bundles over closed
K\"ahler manifolds.
\newblock \textsl{Commun.\ Math.\ Phys.} \textbf{135} (1990) 1--17

\bibitem{CdSLSG}
\textsc{A. Cannas da Silva}: {\it Lectures on Symplectic Geometry}, Springer-Verlag, 2008

\bibitem{DemStu}
\textsc{S. Demoulini {\upshape and} D. Stuart}:
Adiabatic limit and the slow motion of vortices in a Chern--Simons--Schr\"odinger system.
{\sl Commun. Math. Phys.} {\bf 290} (2006) 597--632


\bibitem{determinant}
{\sc P.~Deligne}: {Le d\'eterminant de la cohomologie}.
 \textsl{Contemp. Math.} \textbf{67} (1985) 93--117


\bibitem{Dey}
\textsc{R. Dey}:
Geometric prequantization of the moduli space of the vortex equations on a Riemann surface.
{\sl J. Math. Phys.} {\bf 47} (2009) 103501; Erratum {\sl ibid.} {\bf 50} (2009) 119901


\bibitem{Dir}
\textsc{P.A.M. Dirac}:  {\it The Principles of Quantum Mechanics}, 4th Edition, Clarendon Press, 1958


\bibitem{DonMT}
{\sc S.K.~Donaldson}: 
Topological field theories and formulae of Casson and Meng--Taubes. In: {\sc J. Hass, M. Scharlemann} (Eds.): {\it Proceedings of the Kirbyfest (Berkeley, CA, 1998)}, International Press, 1999; pp. 87Ð102

\bibitem{DonKro}
\textsc{S. Donaldson {\upshape and} P. Kronheimer}: {\it The Geometry of Four-Manifolds},
Clarendon Press, 1990

\bibitem{Elkik}
\textsc{R. Elkik}: Fibr\'es d'intersections et int\'egrales de classes de Chern.
\newblock \textsl{Ann. Sci. \'Ecole Norm. Sup.} {\bf 22} (1989) 195-226

\bibitem{GarVRS}
\textsc{O. Garc\'\i a-Prada}: A direct existence proof for the vortex
equations over a compact Riemann surface.
\newblock \textsl{Bull. London. Math. Soc.} \textbf{26} (1992) 88--96

\bibitem{GarICV}
\textsc{O. Garc\'\i a-Prada}: Invariant connections and vortices.
\newblock \textsl{Commun. Math. Phys.} \textbf{156} (1993) 527--546



\bibitem{GH}
\textsc{P. Griffiths {\upshape and} J. Harris}: {\it Principles of Algebraic Geometry}, John Wiley and Sons, 1978


\bibitem{Hat}
\textsc{A. Hatcher}: {\it Algebraic Topology}.
Cambridge University Press, 2002

\bibitem{Hit}
\textsc{N.J. Hitchin}: Flat connections and geometric quantization.
\newblock \textsl{Commun. Math. Phys.} \textbf{131} (1990) 347--380


\bibitem{JafTau}
\textsc{A. Jaffe {\upshape and} C. Taubes}: {\it Vortices and Monopoles},
Birkh\"auser, 1980


\bibitem{JosCRS}
\textsc{J. Jost}: {\it Compact Riemann Surfaces}, 3rd Edition,
Springer-Verlag, 2006

\bibitem{KemAV}
\textsc{G.R. Kempf}: {\it Algebraic Varieties},
Cambridge University Press, 1993


\bibitem{Knudsen-I}
\textsc{F. Knudsen {\upshape and} D. Mumford}: The projectivity of the moduli space of stable curves I: Preliminaries
on ``det'' and ``div''.
\newblock \textsl{Math.\ Scand.} \textbf{39} (1976) 19--55

\bibitem{KroMro}
\textsc{P. Kronheimer {\upshape and} T. Mrowka}: {\it Monopoles and Three-Manifolds},
Cambridge University Press, 2007



\bibitem{KruSut}
\textsc{S. Krusch {\upshape and} P.M. Sutcliffe}: Schr\"odinger--Chern--Simons vortex dynamics.
\newblock \textsl{Nonlinearity} \textbf{19} (2006) 1515--1534

\bibitem{McD}
\textsc{I.G. Macdonald}: Symmetric products of an algebraic curve.
\newblock \textsl{Topology} \textbf{1} (1962) 319--343

\bibitem{MalMan}
\textsc{R. Maldonado {\upshape and} N.S. Manton}: Analytic vortex
solutions on compact hyperbolic surfaces. \newblock \textsl{J.\ Phys. A: Math. Theor.} \textbf{48} (2015) 245403


\bibitem{ManFVD}
\textsc{N.S. Manton}: First-order vortex dynamics.
\newblock \textsl{Ann. Phys.} \textbf{256} (1997) 114--131

\bibitem{ManNasV}
\textsc{N.S. Manton {\upshape and} S.M. Nasir}: Volume of vortex moduli spaces.
\newblock \textsl{Commun. Math. Phys.} \textbf{199} (1999) 591--604


\bibitem{ManRom}
\textsc{N.S. Manton {\upshape and} N.M. Rom\~ao}: Vortices and
Jacobian varieties. \newblock \textsl{J.\ Geom.\ Phys.} \textbf{61} (2011) 1135--1155

\bibitem{ManSutTS}
\textsc{N. Manton {\upshape and} P. Sutcliffe}: {\it Topological Solitons},
Cambridge University Press, 2004


\bibitem{Mattuck}
\textsc{A. Mattuck}: Symmetric products and Jacobians.  {\sl Amer. J. Math.} \textbf{83} (1961) 189--206

\bibitem{McDSalST}
\textsc{D. McDuff {\upshape and} D. Salamon}: {\it Introduction to Symplectic Topology}, 2nd Edition,
Clarendon Press, 1998


\bibitem{Mol}
\textsc{V.H. Moll}: {\it Numbers and Functions}, American Mathematical Society, 2012

\bibitem{MunHK}
\textsc{I. Mundet i Riera}: A Hitchin--Kobayashi correspondence for K\"ahler fibrations.  {\sl J. reine angew. Math.} \textbf{528} (2000) 41--80


\bibitem{Ngu}
\textsc{T. Nguyen}: Lagrangian correspondences and Donaldson's TQFT construction of the Seiberg-Witten invariants of 3-manifolds.
\newblock \textsl{Alg. Geom. Topol.} \textbf{14} (2014) 863--923



\bibitem{NogYMH}
\textsc{M. Noguchi}: Yang--Mills--Higgs theory on a compact Riemann surface.
\newblock \textsl{J. Math. Phys.} \textbf{28} (1987) 2343--2346




\bibitem{Per}
\textsc{T. Perutz}: Symplectic fibrations and the Abelian vortex equations.
 \newblock \textsl{Commun. Math. Phys.} \textbf{278} (2008) 289--306

\bibitem{Pol}
\textsc{A. Polishchuk}: {\it Abelian Varieties, Theta Functions and the Fourier Transform},
Cambridge University Press, 2003


\bibitem{Qui}
\textsc{D. Quillen}: Determinants of Cauchy--Riemann operators over a Riemann surface.
\newblock \textsl{Funct. Anal. Appl.} \textbf{19} (1985) 31--34

\bibitem{RomPhD}
\textsc{N.M. Rom\~ao}: {\it Classical and Quantum Aspects of Topological Solitons}, PhD thesis, University of Cambridge, 2002

\bibitem{RomQVS}
\textsc{N.M. Rom\~ao}: Quantum Chern--Simons vortices on a sphere.
\newblock \textsl{J. Math. Phys.} \textbf{42} (2001) 3445--3469

\bibitem{RomSpe}
\textsc{N.M. Rom\~ao  {\upshape and}  J.M. Speight}: Slow Schr\"odinger dynamics of gauged vortices.
{\sl Nonlinearity} {\bf 17} (2004) 1337--1355

\bibitem{SamVS}
\textsc{T.M. Samols}: Vortex scattering.
\newblock \textsl{Commun. Math. Phys.} \textbf{145} (1992) 149--180

\bibitem{Seg}
{\sc G.~Segal}: 
Geometric aspects of quantum field theory. In: {\sc I. Satake} (Ed.): {\it Proceedings of the International Congress of Mathematicians (Kyoto, 1990)}, Vol. 2, Springer-Verlag, 1992; pp. 1387--1396

\bibitem{Sou}
\textsc{C. Soul\'e, D. Abramovich, J.-F. Burnol {\upshape and} J. Kramer}: {\it Lectures on Arakelov Geometry},
Cambridge University Press, 1992


\bibitem{Str}
\textsc{I.A.B. Strachan}: Low-velocity scattering of vortices in a modified
Abelian Higgs model.
\newblock \textsl{J. Math. Phys.} \textbf{33} (1992) 102--110

\bibitem{StuAHM}
\textsc{D.M.A. Stuart}: Dynamics of Abelian Higgs vortices in the near
Bogomolny regime.
\newblock \textsl{Commun. Math. Phys.} \textbf{159} (1994) 51--91

\bibitem{StuAL}
\textsc{D.M.A. Stuart}: Analysis of the adiabatic limit for solitons in classical field theory.
\newblock \textsl{Proc. Roy. Soc. Lond. A} \textbf{463} (2007) 2753--2781


\bibitem{TonTur}
\textsc{D. Tong  {\upshape and}  C. Turner}:  Quantum Hall effect in supersymmetric Chern--Simons theories.
{\sl Phys. Rev. B} {\bf 92} (2015) 235123



\bibitem{WooGQ}
\textsc{N.M.J. Woodhouse}: {\it Geometric Quantization}, 2nd Edition, Clarendon Press,
1991



\end{small}
\end{thebibliography}

\end{document}